\documentclass[acmsmall,screen]{acmart}
\AtBeginDocument{%
  }

\renewcommand\footnotetextcopyrightpermission[1]{} %

\usepackage{apxproof}

\usepackage[capitalise,nameinlink]{cleveref} %

\usepackage[ruled,noend,linesnumbered]{algorithm2e} %

\usepackage{setspace} %

\DontPrintSemicolon     %
\SetNlSty{}{}{}                %
\SetAlgoInsideSkip{smallskip}   %

\SetAlFnt{\small}			%
\SetAlCapFnt{\small}		%
\SetAlCapNameFnt{\small}

\newcommand{\algocomment}[1]{\textcolor{teal}{{\tcp{\textnormal{#1}}}}}
\newcommand{\algocommentinline}[1]{\tcp*[f]{\textcolor{teal}{\textnormal{#1}}}} %
\SetKwComment{tcp}{\color{teal}// }{}        %
\Crefname{algocf}{Algorithm}{Algorithms}

\def\e#1{\emph{#1}}

\newcommand{\hide}[1]{}

\newcommand{\var}{\texttt{var}}
\newcommand{\cnt}{\mathtt{cnt}}
\newcommand{\bigO}{\mathcal{O}}
\def\polylog{\operatorname{polylog}} %
\newcommand{\datarule}{{\,:\!\!-\,}}
\newcommand{\sem}{\mathbb{S}}
\newcommand{\Szero}{\bar{0}}
\newcommand{\Sone}{\bar{1}}
\newcommand{\agg}{\texttt{agg}}
\newcommand{\AGG}{\texttt{AGG}}
\newcommand{\val}{\texttt{val}}
\newcommand{\sub}{\texttt{sub}}
\newcommand{\N}{\mathbb{N}}
\newcommand{\dom}{\mathrm{dom}}
\newcommand{\ar}{\texttt{ar}}

\newcommand{\calH}{{\mathcal H}}

\newcommand{\pord}{\pi}
\newcommand{\pordset}{\Pi}
\newcommand{\neighbors}[1]{\text{Neighbors}_{#1}}
\newcommand{\parent}[1]{\text{Parent}_{#1}}
\newcommand{\branchRoots}{\texttt{branchRoots}}
\newcommand{\subtree}[2]{#1[#2]}
\newcommand{\restrict}[2]{#1|_{#2}}
\newcommand{\res}{\texttt{result}}
\newcommand{\resr}[1]{\res_{#1}}
\newcommand{\highest}[1]{\text{root}_{#1}}
\newcommand{\trunktree}{T_{\text{trunk}}}
\newcommand{\trunkpord}{\pord_{\text{trunk}}}
\newcommand{\trunkvars}{X_{\text{trunk}}}

\newcommand{\minval}{\texttt{minVal}}
\newcommand{\qid}{\texttt{Q}_\texttt{ID}}
\newcommand{\qeq}{\texttt{count}}
\newcommand{\qps}{\texttt{smaller}_\texttt{CQ}}
\newcommand{\tps}{\texttt{smaller}_\texttt{total}}
\newcommand{\ps}{\texttt{smaller}}

\newcommand{\hyperclique}{\textsc{Hyperclique}}
\newcommand{\seth}{\textsc{Seth}}
\newcommand{\BMM}{\textsc{BMM}}

\usepackage{xspace}
\newcommand{\MINDA}{\textsc{MinDA\xspace}}

\newcommand{\introparagraph}[1]{\textbf{#1.}} %

\newtheorem{remark}{Remark}

\newtheoremrep{theorem}{Theorem}
\newtheoremrep{lemma}{Lemma}
\newtheoremrep{observation}{Observation}

\usepackage{tikz}
\newcommand*\gcircled[1]{\tikz[baseline=(char.base)]{
            \node[shape=circle,draw,fill=gray!20,inner sep=1pt] (char) {#1};}}

\begin{document}

\author{Nofar Carmeli}
\email{Nofar.Carmeli@inria.fr}
\orcid{0000-0003-0673-5510}
\affiliation{%
  \institution{Inria, LIRMM, University of Montpellier, CNRS}
  \country{France}
}

\author{Nikolaos Tziavelis}
\email{ntziavel@ucsc.edu}
\orcid{0000-0001-8342-2177}
\affiliation{%
  \institution{UC Santa Cruz}
  \country{USA}
}


\title{Fine-Grained Dichotomies for Conjunctive Queries with Minimum or Maximum}

\begin{abstract}

We investigate the fine-grained complexity of direct access to Conjunctive Query (CQ) answers according to their position, ordered by the minimum (or maximum) value between attributes.
We further use the tools we develop to explore a wealth of related tasks. We consider the task of ranked enumeration under min/max orders, as well as tasks concerning CQs with predicates of the form $x \leq \min X$, where $X$ is a set of variables and $x$ is a single variable: counting, enumeration, direct access, and predicate elimination (i.e., transforming the pair of query and database to an equivalent pair without min-predicates).
For each task, we establish a complete dichotomy for self-join-free CQs, precisely identifying the cases that are solvable in near-ideal time, i.e., (quasi)linear preprocessing time followed by constant or logarithmic time per output.
\end{abstract}

\maketitle

\section{Introduction}

\introparagraph{Query-Answering Tasks}
Database systems strive to answer queries efficiently, but
the number of answers can sometimes be very large.
Instead of computing all answers at once, different \emph{query-answering tasks} 
aim to extract useful information about the answers in significantly less time.
\emph{Enumeration}~\cite{bdg:dichotomy,Berkholz20tutorial} is the task of producing answers one-by-one,
focusing on the time for each individual answer rather than the total time.
\emph{Counting}~\cite{pichler13counting} reports only the number of answers.
\emph{Ranked direct access}~\cite{carmeli23tractable} asks for a data structure that allows accessing any answer by its index, as if the answers were stored in a sorted array (but ideally without materializing it).
This task is powerful as it supports all previous tasks as special cases,
while also supporting other use-cases, such as histograms, quantiles, boxplots, or serving ranked answers in pages.
So, it should be no surprise that ranked direct access cannot always be supported efficiently, without first producing all query answers. 
Relaxations include \emph{unranked direct access}, which can be sufficient for sampling (with or without replacement)~\cite{nofar22random},
and \emph{selection} or \emph{quantile computation}~\cite{tziavelis23quantiles}, 
which is the single-access task.

We focus on \emph{Conjunctive Queries} (CQs), a class of queries that is often studied to reason about the complexity of query-answering tasks, as this class is relatively simple, but complex enough to capture the main difficulties of query evaluation, caused by joins.
A first step in understanding the fine-grained complexity is to characterize which problem instances (defined by a query and sometimes an order) can be supported in time close to the trivial lower bound of linear time.
This time bound is specified with respect to the database size, and not the number of query answers which may be much larger, or the query size that we treat as a constant.
Allowing for logarithmic factors, our work focuses on \e{quasilinear preprocessing} time and \emph{logarithmic time per output}; these are the time bounds that we refer to as efficient. 
This boundary is also practically important since polynomially-larger algorithms cannot always run on large-scale databases in reasonable time.

\introparagraph{Min-Ranked Direct Access}
For the task of ranked direct access, a characterization of efficient queries and orders is known for certain families of orders~\cite{carmeli23tractable}.
These are (full or partial) lexicographic orders, specified by an ordering of the free query variables, and orders by a sum of variables, specified by the subset of variables comprising the sum. 
A characterization is also known for the single-access task~\cite{tziavelis23quantiles}, which also extends to a third family of orders that sorts the answers by a \e{minimum (or maximum)} of a set of variables~\cite{tziavelis23quantiles}: any \emph{acyclic free-connex} CQ admits single access with any min/max order in quasilinear time. 
Since direct access has not been considered before for such orders, this leaves an obvious gap in the literature that motivates the first question we address: Which CQs and variable sets admit efficient min-ranked direct access?
Eldar et al.~\cite{eldar2024direct} have considered direct access with aggregates such as minimum,
but their problem is fundamentally different because
the minimum is on one variable between answers, while here the minimum is between variables within one answer.

\begin{example}
\label{ex:min_da}
A large software company wishes to inspect the reliability of teams that may form to handle urgent system failures. Each team comprises three professionals on call in the same shift: a Backend Engineer, a Cloud Specialist, and a Security Analyst. The team reliability is assessed to be the minimum reliability score of its members. If the database contains the relation $\text{Workers}(\text{ID}, \text{role}, \text{reliability}, \text{shift})$, we can use the query that joins $\text{Workers}(ID1, \text{`engineer'}, r1, s)$, $\text{Workers}(ID2, \text{`cloud'}, r2, s)$, and $\text{Workers}(ID3, \text{`security'}, r3, s)$ on the shift and sorts the answers by $\min(r1, r2, r3)$.
Ranked direct access allows gaining insights about the distribution of the reliability scores of possible teams.
This paper shows that it can indeed be supported efficiently.
\end{example}

In each query answer, the value of the minimum function can be attributed to a specific query variable, say $x$, and our algorithm uses the identity of $x$ to split the set of query answers.
To achieve this splitting, we use an extension of CQs that includes a minimum predicate of the form $x\leq \min X$, where $X$ is a set of variables. Beyond direct access, this class of queries can be interesting on its own. 
For the scenario of \Cref{ex:min_da}, adding $r1 \leq \min \{ r2, r3 \}$ to the CQ would give us the teams whose weakest link in terms of reliability is the Backend Engineer.

\introparagraph{Min-Predicate Elimination}
A common strategy for handling predicates that do not correspond to materialized relations (such as our min-predicate) is to transform the query into a different query without such predicates, possibly over a modified database.
This process, which we call \e{predicate elimination}, takes as input a database and returns a new database and a new query with effectively the same set of answers.
Efficient predicate elimination algorithms are known for 
non-equality predicates (sometimes called disequalities)~\cite{papadimitriou99complexity}, 
not-all-equal predicates~\cite{khamis19negation},
predicates of the type $\min{X}=c$ for a constant $c$~\cite{tziavelis23quantiles},
and some cases of inequality predicates~\cite{tziavelis21inequalities}.
Our min-predicate $x\leq \min X$ can be viewed as a
``star-shaped'' conjunction of inequalities $x\leq x'$ for each $x'\in X$,
but eliminating these inequalities with the known elimination algorithm~\cite{tziavelis21inequalities} requires a join tree in which every inequality is subsumed by an edge.\footnote{This condition was also used to evaluate aggregate queries with inequalities~\cite{khamis20faqai}.}
It is not known when such a join tree exists and whether this condition is necessary; an exact characterization of when inequalities can be efficiently eliminated is an open problem.
We resolve this problem specifically for the inequalities that correspond to a min-predicate.
The focus of our algorithm is to bring the query into a form that allows the inequalities to be eliminated by the known algorithm.

The queries obtained from predicate elimination are usually designed to belong to a query class known for its efficiency, such as \emph{full acyclic} CQs.
As we attempted to capture all the cases where min-ranked direct access is efficient, we noticed that we could not always find an elimination into a single CQ; instead, we sometimes have to construct a set of CQs, on top of which direct access can work if the CQs have disjoint sets of answers.
This transition from considering a single CQ to a set of disjoint CQs allows us to support more orders and complete the dichotomy.
In this sense, supporting minimum orders turned out to be more technically challenging than supporting lexicographic or sum orders, where each query can be supported using a single \emph{join tree}~\cite{carmeli23tractable}.

\introparagraph{Direct Access Dichotomy}
Using min-predicate elimination for upper bounds and counting-based arguments for lower bounds, we arrive at the structural condition required for efficient min-ranked direct access: 
The query must be acyclic free-connex (which is necessary even for unranked direct access or single access) and, in addition,
there must be no long chordless path between the variables participating in the order.
This condition forms a dichotomy for CQs with no self-joins, where the lower bound assumes the hardness of \emph{hyperclique detection} and the \emph{Strong Exponential Time Hypothesis}, both used before in similar contexts~\cite{mengel2025lower}.
With the dichotomy for min-rankings complete, we observe that, similarly to lexicographic and summation-based orders~\cite{carmeli23tractable,tziavelis23quantiles},
the efficient class of queries and orders for direct access is strictly smaller than the efficient class for single access.

\introparagraph{Elimination Dichotomy}
Predicate elimination is a broadly applicable tool that enables us to carry over algorithms developed for the more commonly studied class of CQs to the CQ with the predicate.
Consequently, we consider it as a problem on its own, and establish a dichotomy for the queries and min-predicates that can be eliminated efficiently.
As it turns out, the condition for efficiency is essentially the same as that of min-ranked direct access.

This urges us to further understand the generality of the elimination approach and the corresponding structural condition.
In order to capture all efficient  cases,
is it enough to consider algorithms that use min-predicate elimination, or are task-specific algorithms required? 
Are all cases with a long chordless path between variables in the minimum function inefficient?
By studying other tasks individually, we find that the answer depends on the task. 
The elimination condition forms a dichotomy for direct access and counting with a min-predicate, but not for enumeration.

\begin{figure}
    \centering
    \includegraphics[width=0.9\linewidth]{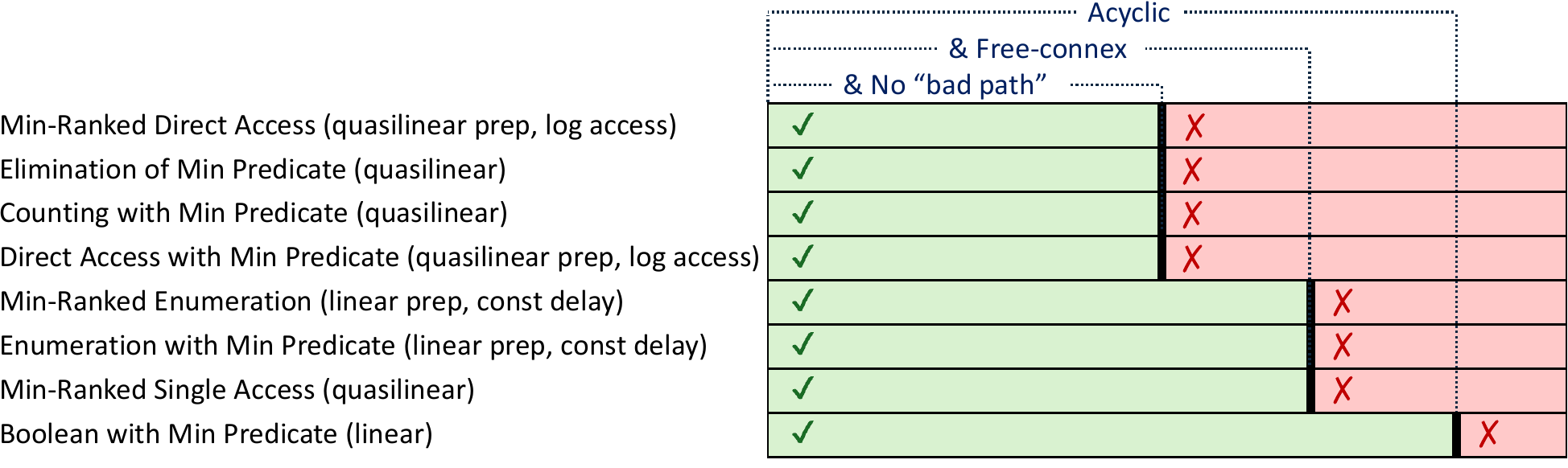}
    \caption{Overview of dichotomies for query-answering tasks. All results are established in this paper, except for min-ranked single access, which follows from results on quantile queries~\cite{tziavelis23quantiles}, as explained in \Cref{sec:quantiles_vs_selection}.  A ``bad path'' is a chordless $k$-path with $ k\geq 3$, which in the case of a problem with min-ranking is between two variables participating in the order, while in the case of a problem with a min-predicate $x \leq \min X$ is between $x$ if it is free and another free variable in $X$. The negative sides of all dichotomies apply only to self-join-free queries and are conditional on hardness hypotheses; see \Cref{sec:results} for detailed statements.}
    \label{fig:overview}
\end{figure}

\introparagraph{Enumeration Dichotomies}
We inspect two enumeration tasks: min-ranked enumeration for a CQ, and unranked enumeration for a CQ with a min-predicate. For both tasks, we find that not all efficiently-solvable cases fit the restrictive condition of predicate elimination. 
In fact, we give linear preprocessing and \emph{constant delay} algorithms that support any min-ranking or min-predicate for all CQs that support efficient enumeration without the min conditions (i.e., acyclic free-connex CQs~\cite{bdg:dichotomy}).
This is perhaps surprising, as this is not the case for other families of orders; as an example, 
while all lexicographic orders can be supported with logarithmic delay on tractable CQs~\cite{tziavelis25anyk},
only some orders are known to be supported with constant delay~\cite{Bakibayev13fdb,braultbaron13thesis}.
Viewing min-predicates as a ``star-shaped'' conjunction of inequalities, we can compare our enumeration algorithm for min-predicates with what is known for enumeration with inequalities. While there is no known dichotomy for a general conjunction, Wang and Yi~\cite{wang2022comparisons} propose an algorithm that applies when the inequalities satisfy a certain acyclicity condition. Our algorithm supports more cases of min-predicates (not satisfying their acyclicity condition), while the algorithm of Wang and Yi supports other cases that do not correspond to min-predicates, hence are not covered by our work.

\introparagraph{Organization}
\Cref{sec:prelims} introduces notation and background results. \Cref{sec:results} formalizes our main theorems. 
\Cref{sec:elimination} presents the predicate elimination algorithm, and \Cref{sec:da_min_alg} the min-ranked direct access algorithm, while \Cref{sec:enum} addresses enumeration.
Conditional lower bounds for predicate elimination, counting, and direct access are given in \Cref{sec:hardness}. We conclude in \Cref{sec:conclusion}.

\begin{toappendix}
\section{Min-Ranked Single Access}
\label{sec:quantiles_vs_selection}

The \emph{single access} problem for a query $Q$ and order $\preceq$, also known as \emph{selection}~\cite{carmeli23tractable}, is the following:
Given a database $D$ and an index $k$ as input, return the $k$-th answer in the list $Q(D)$ sorted by a total order $\preceq$. If $k$ is too large, an ``out-of-bounds'' message is returned.

\begin{theorem}[\cite{tziavelis23quantiles,carmeli23tractable}]
    Let $Q$ be a CQ with no self-joins and $X$ a subset of its free variables. 
    $Q$ admits single access according to $\min{X}$ with quasilinear preprocessing time and logarithmic access time if and only if $Q$ is acyclic free-connex, assuming \hyperclique{} and \seth{}.
\end{theorem}

\begin{proof}
The negative side holds regardless of the order~\cite[Lemma 6.4]{carmeli23tractable}. Intractability follows from the hardness of counting by binary searching for the smallest index that gives an out-of-bound message.
The positive side was shown for the closely related problem of quantile queries.
A quantile query takes as input a $\phi \in [0,1]$ and returns the $\phi$-quantile answer, i.e., the answer at position $\lceil \phi |Q(D)| \rceil$. When the query is free-connex acyclic, counting can be done in linear time, and the two problems coincide. Specifically, to perform single access with the index $k$, it is enough to ask for the quantile $\frac{k}{|Q(D)|}$.
We can restrict the problem to its free variables using \Cref{lem:cq-free-restrict} and then use the algorithm of Tziavelis et al.~\cite[Theorem 5.3]{tziavelis23quantiles} for full acyclic quantile queries with min rankings in quasilinear time. 
\end{proof}

We note that for quantile queries, it is unclear how to transfer the negative result concerning acyclic, but not free-connex CQs.
This is because the lower bound for single access relies on the hardness of counting the number of query answers~\cite{carmeli23tractable};
while single access allows us to compute this count via binary search (i.e., identifying the end of the array), this does not apply to quantile queries.
\end{toappendix}

\section{Preliminaries}\label{sec:prelims}

\subsection{Databases and Queries}

We use
$[m]$ to denote the set of integers $\{1, \ldots, m\}$.

\introparagraph{Databases}
A \emph{schema} is a set of relational symbols $\{ R_1, \ldots, R_m \}$.
A \emph{database} $D$ 
consists of a finite relation $R^D \subseteq \dom(D)^{\ar(R)}$ for each symbol $R$ in the schema,
where $\ar(R)$ is the arity of $R$, and the \emph{domain} $\dom(D)$ is a set of constant values that we assume to be $\N$.
When $D$ is unambiguous, we simply use $R$ instead of $R^D$.
We refer to the elements of $R$ as database \emph{tuples} or \emph{facts}.

\introparagraph{Queries}
A \emph{Conjunctive Query} (CQ) $Q$ is an expression of the form
$Q(\vec{y}) \datarule R_1(\vec{x_1}), \ldots, R_\ell(\vec{x_\ell})$,
where $\{ R_1, \ldots, R_\ell \}$ are relational symbols from a schema, and
$\vec{x_i}, i \in [\ell]$ is a list of $\ar(R_i)$ variables.
The body of the CQ is the logical expression on the right-hand side of the $\datarule$ operator.
Each $R_i(\vec{x_i})$ of the body is called an \e{atom} of $Q$, and we use $\var(a)$ to denote the variables in atom $a$.
The variables $\vec{y}$ are called \emph{free}, and each one must necessarily appear in some $\vec{x_i}, i \in [\ell]$.
The variables that are not free are called \emph{existential}.
A \emph{CQ with predicates} $Q \wedge P$ is a CQ $Q$ with the predicate $P$ added as a conjunction to its body.
The predicate $P$ can use variables of $Q$ and constants, but it does not correspond to any relation in the schema (thus it is not considered materialized). We also assume it can contain arbitrary conjunctions or disjunctions of atomic predicates.
In this work, we focus on min-predicates $P$ of the type $x \leq \min X$, where $x$ is a single variable and $X$ is a set of variables, all from $Q$. 
We note that it is equivalent to $x = \min X$ if $x \in X$.

A \e{homomorphism} from a CQ $Q$ to a database $D$ 
is a mapping of all $Q$ variables to constants from $\dom(D)$,
such that every atom of $Q$ maps to a tuple in $D$.
A \e{query answer} is such a homomorphism 
followed by a projection 
on
the free variables.
The set of query answers is $Q(D)$ and we use $q \in Q(D)$ for a query answer.
These definitions extend in a straightforward way to a CQ with a predicate $P$, in which case the homomorphism must also satisfy $P$.
We denote by $\cnt(Q(D) | P)$
the number of query answers to $Q$ over $D$ that satisfy $P$.

\introparagraph{Hypergraphs}
Two vertices in a hypergraph $\calH(V, E)$ are {\em neighbors} if they appear in the same edge.
A {\em path} is a sequence of vertices such that every two succeeding vertices are neighbors.
A path of length $k$, also referred to as a $k$-path, consists of $k + 1$ vertices.
A {\em chordless path} is a path in which no two non-succeeding vertices appear in the same hyperedge 
(in particular, no vertex appears twice).
Any path $v_1,v_2,\ldots,v_k$ can be contracted into a chordless path between $v_1$ and $v_k$, of possibly smaller length. Starting with $i=1$, we can connect $v_i$ to the largest $v_j$ among its neighbors, remove all variables in-between, and continue with $i=j$.
A \emph{join tree} of a hypergraph is a tree $T$ where the nodes correspond to hyperedges
and for all $u \in V$ the set $\{e \in E \mid u \in e\}$ forms a (connected) subtree in $T$.
A \emph{rooted join tree} is a join tree with a specific node delineated as the root of the tree. 
A hypergraph is (alpha)-acyclic if it admits a join tree.
We also allow for a more relaxed join tree notion, with additional nodes that correspond to restrictions of existing hyperedges, i.e., containing a subset of the vertices.
In particular, we can always introduce a join tree node that corresponds to a single vertex without affecting acyclicity.

\introparagraph{Classes of CQs}
A CQ is \e{full} if it has no existential variables and \e{Boolean} if all variables are existential.
A repeated occurrence of a relational symbol is called a \e{self-join} and if no self-joins exist, a CQ is called \e{self-join-free}.
To determine the acyclicity of a CQ $Q$, we associate a hypergraph $\calH(Q) = (V, E)$ with it,
where the vertices are the variables of $Q$, 
and every atom of $Q$ corresponds to a hyperedge with the same set of variables.
With a slight abuse of notation, we may sometimes identify atoms of $Q$ with hyperedges of $\calH(Q)$.
A CQ $Q$ is {\em acyclic} if $\calH(Q)$ is acyclic,
otherwise it is \e{cyclic}.
Further, it is called free-connex acyclic if it is acyclic and it remains acyclic after the addition of a hyperedge containing exactly the free variables~\cite{braultbaron13thesis,Berkholz20tutorial}.

\introparagraph{Orders over Query Answers}
The query answers $Q(D)$ can be ranked according to a given \e{total order} $\preceq$.
We mainly focus on min-ranking and max-ranking, where each answer is compared according to the value $\min X$ or $\max X$ for a given set $X$ of query variables.
We note that this ranking can be viewed as a totally ordered commutative monoid (where the monoid operator is $\min$ or $\max$) or as a min-max semiring~\cite{tziavelis25anyk}.
We will also make use of (partial) lexicographic orders of one variable $\langle x \rangle$, which simply means that the answers are ordered according to the $x$ value.

\subsection{Query-Answering Tasks}

We define several computational problems for a query $Q$ that can be either a simple CQ or a CQ with predicates. In all cases, the database $D$ is considered the input, while the query $Q$ and any order $\preceq$ over query answers are considered part of the problem definition.

\gcircled{1} \e{Boolean}: Report whether $Q(D)$ is empty.
\gcircled{2} \e{Counting}: Report $|Q(D)|$.
\gcircled{3} \e{(Unranked) Enumeration}: After a preprocessing phase, report $Q(D)$ one-at-a-time without duplicates in arbitrary order.
\gcircled{4} \e{Ranked Enumeration}: After a preprocessing phase, report $Q(D)$ one-at-a-time without duplicates in $\preceq$ order.
\gcircled{5} \e{(Unranked) Direct Access}: Build a data structure that supports accessing any index of a list containing $Q(D)$ in arbitrary order.
\gcircled{6} \e{Ranked Direct Access}: Build a data structure that supports accessing any index of the list $Q(D)$ sorted by $\preceq$ order.

Since $Q$ is considered fixed, we are interested in the data complexity~\cite{DBLP:conf/stoc/Vardi82} of these tasks, i.e., the complexity with respect to $|D|$.
We use the RAM model of computation with uniform cost measure. It allows for linear time construction of lookup tables which can be accessed in 
constant time.
It also allows sorting input relations in linear time~\cite{grandjean1996sorting}, a fact that we use to get linear preprocessing in our enumeration results.

\subsection{Existential Variable Elimination}

Given a CQ, we define its restriction to free variables to be the full CQ obtained by replacing every atom $R(\vec{x})$ with an atom $R'(\vec{x'})$, where $\vec{x'}$ is $\vec{x}$ restricted to the free variables, and $R'$ is a fresh relational symbol that appears only once in the query. 
It is known that acyclic free-connex CQs can be efficiently restricted to their free variables~\cite{Berkholz20tutorial}.
This can be done using the Yannakakis semijoin reduction~\cite{Yannakakis} on a suitable join tree.
The resulting query is full, acyclic, self-join-free, and all chordless paths of free variables in the original query are preserved.

\begin{lemma}\label{lem:cq-free-restrict}[Existential Elimination]
    Let $Q$ be an acyclic free-connex CQ and $Q'$ its restriction to free variables. Given a database $D$, we can build in linear time a database $D'$ such that $Q(D)=Q'(D')$.
\end{lemma}

\subsection{Aggregates for Full Acyclic CQs}

Given a subset of the relations referenced in a full CQ, a partial answer is a set of tuples, one from each relation in the subset, that join with each other, in the sense that they agree on the common variables.
Fixing a rooted join tree $T$ for a full acyclic CQ and materializing a fresh relation for each node in the join tree (which eliminates self-joins), 
let $R$ be one of these relations and $t \in R$ one of its tuples.
We define the set of partial answers in the subtree $\sub_T(t)$ to be all the partial answers that include $t$ from $R$, as well as tuples from relations in the subtree of $R$.

Let $\AGG$ be the following computational problem for a CQ $Q$ that is self-join-free and full acyclic:
We are given as input a database $D$, a rooted join tree $T$ of $Q$, a function $\val(t)$ that associates a value to each tuple $t$ in $D$, and two 
binary, associative, and commutative operators $\oplus, \otimes$.
The output is an aggregate value $\agg(t) = \bigoplus_{p \in \sub_T(t)} \bigotimes_{t' 
\in p} \val(t')$ for each tuple $t$ in $D$.

It is known that the problem can be solved efficiently when the operators form a commutative semiring~\cite{GondranMinoux:2008:Semirings}.
A commutative \emph{semiring} $\sem$ is a 5-tuple $(A, \oplus, \otimes, \Szero, \Sone)$, where $A$ is a non-empty set, $\oplus$ and $\otimes$ are binary, associative, and commutative operators over $A$, $\Szero$ is a neutral element for $\oplus$ (i.e., $x \oplus \Szero = x$, for all $x \in A$), $\Sone$ is a neutral element for $\otimes$ (i.e., $x \otimes \Sone = \Sone \otimes x = x$, for all $x \in A$), the operator $\otimes$ distributes over $\oplus$, and $\Szero$ is absorbing for $\otimes$ (i.e., $x \otimes \Szero = \Szero \otimes x = \Szero$, for all $x \in A$).
The algorithm works in a bottom-up fashion on the join tree and is a generalization of the Yannakakis algorithm~\cite{Yannakakis}.
We give more details in ~\Cref{sec:semiring-agg}.

\begin{lemma}~\cite{abokhamis16faq,Joglekar16ajar}
\label{lem:semiring-agg}
    The problem $\AGG$ is in $\bigO(|D|)$ for full acyclic CQs and semiring aggregates, i.e.,
    if $(A, \oplus, \otimes, \Szero, \Sone)$ is a commutative semiring and $\val$ assigns values from domain $A$.
\end{lemma}

An example use of
\Cref{lem:semiring-agg} is for counting.
We set $\val(t)=1$ for all tuples $t$,
$\oplus$ to sum $(+)$,
and $\otimes$ to product $(\times)$.
The aggregate value of a tuple in the root relation corresponds to the number of CQ answers that agree with the tuple.
To get the final count, we sum these root-relation counts.

\subsection{Hypotheses}

Our lower bounds are conditioned on hardness hypotheses for other problems:
\gcircled{1} \BMM{}~\cite{Berkholz20tutorial}: Two Boolean matrices $A$ and $B$, represented as lists of their non-zero entries, cannot be multiplied in time $m^{1+o(1)}$, where $m$ is the number of non-zeros in $A$, $B$, and $AB$.
\gcircled{2} \hyperclique~\cite{abboud14conjectures,DBLP:conf/soda/LincolnWW18}:
A $(k{+}1,k)$-hyperclique is a set of $k{+}1$ vertices
such that every subset of $k$ vertices is a hyperedge.
For every $k \geq 2$, there is no
$O(m \polylog m)$ algorithm for deciding the existence of a
$(k{+}1,k)$-hyperclique in a hypergraph with $m$ hyperedges.
\gcircled{3} \seth~\cite{impagliazzo01seth}:
    For the satisfiability problem with $m$ variables and 
    $k$ variables per clause ($k$-SAT),
    if $s_k$ is the infimum of the real numbers $\delta$
    for which $k$-SAT admits an
    $\bigO(2^{\delta m})$ algorithm,
    then $\lim_{k \to \infty} s_k = 1$.

\begin{toappendix}
\section{Algorithm for Bottom-up Aggregate Computation}
\label{sec:semiring-agg}

We give more details about the linear time algorithm of \Cref{lem:semiring-agg}.
Recall that in the $\AGG$ problem, the goal is to compute an aggregate value $\agg(t) = \bigoplus_{q_T(t) \in \sub_T(t)} \bigotimes_{t' 
\in q_T(t)} \val(t')$ for each tuple $t$ in a given database $D$, 
where $\val(t)$ is a given assigned value.

\introparagraph{Preprocessing}
Recall that a rooted join tree $T$ has been given as input.
We materialize a distinct relation for every $T$-node.
For every parent node $V_p$ and child node $V_c$, we group the $V_c$-relation
by the $V_p \cap V_c$ variables.
We refer to these groups of tuples as \emph{join buckets}; a join bucket shares the same 
values for variables that appear in the parent node.

\introparagraph{Bottom-up Messages}
The algorithm visits the relations in a bottom-up order of $T$,
sending children-to-parent messages.
As we traverse the relations in bottom-up order, every tuple $t$ computes its $\agg(t)$ and emits it as a message upwards.
The messages are aggregated as follows:

\begin{enumerate}
    \item Messages emitted by tuples $t'$ in a join bucket are aggregated with the operator $\oplus$.
    The result is sent to all parent-relation tuples that agree with the join values of the bucket. 
    
    \item A tuple $t$ computes $\agg(t)$ by aggregating the messages received from all children in the join tree, together with the initial value of $\val(t)$, with the operator $\otimes$.    
\end{enumerate}

The algorithm takes linear time because (1) every tuple is visited once and (2) the aggregation within a join bucket is linear, and (3) the join bucket partition the tuples of each relation.
Correctness is guaranteed because of the distributivity property of the semiring.
\end{toappendix}

\section{Main Results}\label{sec:results}

In this section, we state our main results.
In all cases, the restriction to self-join-free queries and the fine-grained complexity hypotheses are required only for the negative results. In addition, all results also hold for the symmetric case of max.

We start by identifying the cases in which we can efficiently transform a CQ with a min-predicate into an equivalent set of disjoint CQs.
In particular, we would like the CQs in the set to be full and acyclic, as such CQs are known to admit efficient algorithms for a variety of tasks.

\begin{definition}[Predicate Elimination]
\label{def:pred-elimination}
Let $Q$ be a CQ and $P(X)$ a predicate with variables $X \subseteq \var(Q)$.
An \emph{elimination} of $P(X)$ from $Q$ is the computational problem that takes a database $D$ as input, and returns a set of $\ell$ pairs of CQs and databases $\{(Q_i, D_i) \,|\, i \in [\ell]\}$ such that:
\begin{itemize}
    \item The number of queries $\ell$ and the query sizes $|Q_i|$ do not depend on $D$, for all $i \in [\ell]$,
    \item $\var(Q)\subseteq\var(Q_i)$, for all $i \in [\ell]$,
    \item projecting to $\var(Q)$ is a bijection from $\cup_{i \in [\ell]}{Q_i(D_i)}$ to $(Q \wedge P)(D)$, and
    \item for each answer to $(Q \wedge P)(D)$, there is a unique $i \in [\ell]$ such that the answer appears in $Q_i(D_i)$ projected to $\var(Q)$.
\end{itemize}
We say that the elimination is full and acyclic if all CQs $Q_i$ are full and acyclic.\footnote{The term ``predicate trimming'' was used before~\cite{tziavelis23quantiles} to refer to a similar problem, but there the resulting query is a CQ and the bijection between the answer sets is not required to be projection.}
\end{definition}

\begin{theoremrep}[Predicate Elimination]\label{thm:remove-min}
    Let $Q$ be a self-join-free CQ, $X$ a subset of its variables, and $P \equiv (x \leq \min X)$ a predicate.  
    A full acyclic elimination of $P$ from $Q$ is possible in quasilinear time if and only if $Q$ is acyclic free-connex and it is not the case that $x$ is free and there is a chordless $k$-path between $x$ and another free variable in $X$ with $k\geq 3$, assuming \hyperclique{} and \seth{}.
\end{theoremrep}
\begin{appendixproof}

    For the positive side, the first step is to use \Cref{lem:restrict-predicate-to-free} to restrict the query to its free variables.
    So, we can now assume that $Q$ is a full acyclic CQ.

    Next, we transform the domains of variables so that there
    is never an equality between the assignments to different variables.
    This allows us to work with strict inequalities and strict total orders. 
    First, we make sure the query is self-join-free by duplicating relations if needed. Now we can transform the domain of each variable separately to ensure they have disjoint domains.
    For a variable $v$, we replace the domain element $c$ with the tuple $(c,v)$. This construction preserves the relative order between different assignments to variables. That is, if $x_1$ is assigned $c_1$ and $x_2$ is assigned $c_2$ with $c_1<c_2$, over this transformation we also get that $(c_1,x_1)<(c_2,x_2)$ regardless of the order between $x_1$ and $x_2$.
    We have to be careful about how we break ties.
    Since we handle a predicate of the form $x \leq \min X$, we set $x_0$ to have the smallest identifier. This way, we have that $(c,x_0)<(c,x')$  for all $x'\in X$, and a solution that assigns $c$ to both variables will satisfy the predicate $x_0<x'$. Thus, to enforce the predicate $x \leq \min X$, we can simply consider all strict orders that enforce all inequalities of the form $x_0<x'$ with $x'\in X$.\footnote{\label{footnote:other-predicate}If we want to handle a predicate of the form $x_0<\min X$ with $x_0\notin X$ instead, we can set $x_0$ to have the largest identifier. This way, we have that $(c,x_0)>(c,x')$ for all $x'\in X$ and a solution that assigns $a$ to both variables will not satisfy the predicate $x_0<x'$.}
    
    Now we can use \Cref{lem:partitioning} with $Q'$, $x_0$ and $X'$ to find a collection of strict partial orders $\pord_1,\ldots,\pord_k$ that form a partition of $P$ for the case of disjoint domains. Due to \Cref{cor:enforce_order}, each of these orders independently can be eliminated.
    Let $Q_i$ and $D_i$ be the result of eliminating $P_{\pord_i}$ from $Q'$ given $D'$.
    Rename the relations and the fresh variables in the eliminations such that for $i\neq j$ we have that $Q_i$ and $Q_j$ use disjoint relation names and $\var(Q_i)\setminus\var(Q)$ and $\var(Q_j)\setminus\var(Q)$ are disjoint. Set $Q''$ to be the set containing $Q_i$ for $i\in[k]$, and set $D''$ to be the union of $D_i$ for $i\in[k]$.
    Since databases $D_i$, $D_j$ with $i\neq j$ use different relations, we have that $Q_i(D_i)=Q_i(D'')$. Since the strict partial orders are disjoint, we have that the queries in the set give disjoint answer sets. Thus $Q''$ and $D''$ form the required elimination.

    For the negative side,
    if a full acyclic elimination existed, we would be able to count for it in linear time (by summing the count for the individual CQs in the set), contradicting \Cref{lem:counting-hard-side} that claims it is not possible in the case described here (assuming \hyperclique{} and \seth{}).
\end{appendixproof}

\begin{remark}[Elimination Composition]\label{remark:composition}
Predicate elimination is not composable in general; it can be possible to eliminate two predicates independently, but not their conjunction. Eliminating one min-predicate can introduce a long chordless path, preventing the elimination of another. For example, consider
$Q(x_1, x_2, x_3, y) \datarule R_1(x_1), R_2(x_2, y), R_3(x_3, y)$ with predicates $x_1 \leq \min(x_1, x_2)$ and $x_1 \leq \min(x_1, x_3)$.
By \Cref{thm:remove-min}, either predicate can be removed. However, eliminating the first introduces a fresh variable into the first two atoms, connecting $x_1$ and $x_3$ with a chordless path of length 3, thereby preventing the elimination of the second. We will revisit a similar query in \Cref{ex:partitioning-order}.
\end{remark}

\begin{remark}[Choice of Predicate]\label{rem:predicate-choice}
    The semantics of $x \leq \min X$ is that $x$ is less than or equal to all $X$ variables, regardless of whether $x \in X$.
    Our results extend easily to the predicate that requires $x$ to be strictly the largest, written as $\min X > x$ with $x \notin X$.
\end{remark}

To prove the negative side of the elimination dichotomy, we consider the problem of counting for CQs with min-predicates.

\begin{theoremrep}[Counting With a Predicate]\label{thm:count-min-dichotomy}
    Let $Q$ be a self-join-free CQ, $X$ a subset of its variables, and $P \equiv (x \leq \min X)$ a predicate.  
    Counting the answers of $Q\wedge P$ is possible in quasilinear time if and only if $Q$ is acyclic free-connex and it is not the case that $x$ is free and there is a chordless $k$-path between $x$ and another free variable in $X$ with $k\geq 3$, assuming \hyperclique{} and \seth{}. 
\end{theoremrep}
\begin{appendixproof}
    The positive side is a consequence of our predicate elimination result: we can eliminate the predicate according to \Cref{thm:remove-min}, then apply the known counting algorithm that takes linear time on each CQ in the obtained set, and we can simply sum up the results because the CQs in the set have disjoint answer sets.
    The negative side is given in \Cref{lem:counting-hard-side}.
\end{appendixproof}

We prove the positive side of \Cref{thm:remove-min} in \Cref{sec:elimination}, and the negative side of \Cref{thm:count-min-dichotomy} in \Cref{sec:count-neg}.
As counting can be done efficiently for full acyclic CQs, by summing up the counts for each CQ in the set, efficient predicate elimination implies efficient counting. Thus, the positive side of \Cref{thm:remove-min} implies the positive side of \Cref{thm:count-min-dichotomy}, and the negative side of \Cref{thm:count-min-dichotomy} implies the negative side of \Cref{thm:remove-min}.
A noteworthy private case of \Cref{thm:count-min-dichotomy} is the Boolean task.

\begin{corollary}[Boolean With a Predicate]
    Let $Q$ be a self-join-free CQ, $X$ a subset of its variables, and $P \equiv (x \leq \min X)$ a predicate.  
    The Boolean task for $Q\wedge P$ is possible in quasilinear time if and only if $Q$ is acyclic, assuming \hyperclique{} and \seth{}. 
\end{corollary}

Using the predicate elimination technique, we devise an algorithm for direct access with a min ranking function in \Cref{sec:da_min_alg}.
We match this algorithm with a lower bound in \Cref{sec:neg-da}.

\begin{theoremrep}[Ranked Direct Access]\label{thm:min-da}
    Let $Q$ be a CQ with no self-joins and $X$ a subset of its free variables. 
    $Q$ admits direct access according to $\min{X}$ with quasilinear preprocessing time and logarithmic access time if and only if $Q$ is acyclic free-connex and has no chordless $k$-path between two variables in $X$ with $k\geq 3$, assuming \hyperclique{} and \seth{}.
\end{theoremrep}
\begin{appendixproof}
    For the positive side, we first use \Cref{lem:cq-free-restrict} to restrict $Q$ to its free variables, and then use \Cref{thm:ranked-da-pos}. The negative side is given by \Cref{lem:DA-hard-side}.
\end{appendixproof}

Another consequence of our results is a dichotomy for direct access with minimum predicates under arbitrary order. The positive side can be achieved using predicate elimination, and the negative side is a consequence of the hardness of counting.

\begin{theoremrep}[Direct Access With a Predicate]\label{thm:da-pred}
    Let $Q$ be a self-join-free CQ, $X$ a subset of its variables, and $P \equiv (x \leq \min X)$ a predicate.  $Q\wedge P$ admits direct access (with arbitrary order) with quasilinear preprocessing time and logarithmic access time if and only if $Q$ is acyclic free-connex and it is not the case that $x$ is free and there is a chordless $k$-path between $x$ and another free variable in $X$ with $k\geq 3$, assuming \hyperclique{} and \seth{}. 
\end{theoremrep}
\begin{appendixproof}
    The positive side resembles our proof of \Cref{thm:ranked-da-pos}, but is much simpler.
    We first apply our predicate elimination result (\Cref{thm:remove-min}) to obtain a set of full acyclic CQs.
    Each CQ in the set supports direct access (for an arbirary order) with quasilinear preprocessing and logarithmic access time~\cite{nofar22random}.
    To order the answers between the queries in the set consistently, we order them by the position of the query in the set (e.g., the answers to the first query are ordered before the answers to the second query).
    In the preprocessing phase, we compute and store the answer count of each query, which is possible in linear time because the queries are full acyclic~\cite{mengel2025lower}.
    We get pairs of the form $(\qid,\qeq)$ that we extend by a prefix sum to an array of triples $(\qid,\qeq,\ps)$, where $\ps$ is the sum of $\qeq$ of all previous entries.
    In the access phase, when accessing a target index $k$, we binary search for the largest entry $i$ such that $k\geq \ps(i)$. Then, we access answer number $k-\ps(i)$ in the query $\qid(i)$ and return its projection to the variables of $Q$.
    
    The negative side is a consequence of our counting result: if we had direct access with quasilinear preprocessing time and polylogarithmic access, we would be able to binary search for the number of answers, and count in quasilinear time; this is not possible according to \Cref{thm:count-min-dichotomy}.
\end{appendixproof}

Unlike counting and direct access, we show dichotomies for enumeration tasks where the tractable cases are wider than those achievable with predicate elimination. We devise a direct algorithm for ranked enumeration in \Cref{sec:ranked-enum}, and a direct algorithm for enumeration with a min-predicate in \Cref{sec:pred-enum}. A lower bound for the cases not covered by these algorithms is a consequence of the known hardness of enumeration in the case with no ranking or predicates~\cite{bdg:dichotomy,braultbaron13thesis}.

\begin{theoremrep}[Ranked Enumeration]\label{thm:ranked-enum-dichotomy}
    Let $Q$ be a self-join-free CQ and $X$ a subset of its free variables. 
    $Q$ admits (ranked) enumeration according to $\min{X}$ with linear preprocessing time and constant delay if and only if $Q$ is acyclic free-connex, assuming \hyperclique{} and \BMM{}.
\end{theoremrep}
\begin{appendixproof}
    For the positive side, we use \Cref{lem:cq-free-restrict} to restrict the problem to its free variables and then apply \Cref{lem:ranked-enum-pos} on the resulting full acyclic CQ.
    The negative side follows from the hardness of enumeration without order requirements: The answers to a self-join-free CQ that is not acyclic free-connex cannot be enumerated with linear preprocessing and constant delay, assuming \hyperclique{} and \BMM{}~\cite{bdg:dichotomy,braultbaron13thesis}. 
\end{appendixproof}

\begin{theoremrep}[Enumeration With a Predicate]\label{thm:enum-pred}
    Let $Q$ be a self-join-free CQ, $X$ a subset of its variables, and $P \equiv (x \leq \min X)$ a predicate. Enumerating the answers of $Q\wedge P$ is possible with linear preprocessing and constant delay iff $Q$ is acyclic free-connex, assuming \hyperclique{} and \BMM{}.
\end{theoremrep}
\begin{appendixproof}
    For the positive side, use \Cref{lem:restrict-predicate-to-free} to restrict the problem to its free variables, and then apply \Cref{lem:enum-predicate-pos}.

    For the negative side, the answers to a self-join-free CQ that is not acyclic free-connex cannot be enumerated with linear preprocessing and constant delay, assuming \hyperclique{} and \BMM{}~\cite{bdg:dichotomy,braultbaron13thesis}. We can set the variables domain such that $x_0$ is always smaller than the other variables and the predicate always hold. Since $Q$ cannot be efficiently enumerated if it is not acyclic free-connex, neither can $Q\wedge P$.
\end{appendixproof}

\section{Predicate Elimination Algorithm}\label{sec:elimination}

In this section, we prove the positive case of \Cref{thm:remove-min},
showing that if the conditions are satisfied,
an efficient elimination of $x_0 \leq \min X$ exists.
Without loss of generality, we assume $x_0\notin X$.
We first consider full queries, and then handle projections in \Cref{sec:elimin-ext}.

\subsection{Predicate Elimination for Full Queries}

Our strategy rewrites the predicate $x_0 \leq \min X$
as a disjunction of partial orders $\pord$ on $X$ where $x_0$ is the minimum element.
These orders are disjoint in the sense that any variable assignment satisfies exactly one of them.
Each order $\pord$ can also be viewed as a conjunction of inequalities, which we denote as a predicate $P_\pord$.
Such a conjunction can be eliminated as long as there exists a join tree
such that the variables of each inequality appear in adjacent nodes~\cite{tziavelis21inequalities}.
Thus, if we guarantee that this condition holds for our chosen partial orders, we can successfully eliminate them. 
The following example shows how our elimination strategy works end-to-end.

\begin{example}\label{ex:partitioning-order}
We want to eliminate $x_0 \leq \min(x_1, x_2)$ from $Q(x_0, x_1, x_2, y) \allowbreak \datarule \allowbreak R_0(x_0), \allowbreak R_1(x_1, y), \allowbreak R_2(x_2, y)$.
This is in essence the same example as in \Cref{remark:composition}.
Assuming that no two variables can be assigned equal values (an easily enforceable assumption by creating disjoint domains for different variables), there are two possible cases where $x_0$ is the minimum:
Either $x_0 < x_1 < x_2$ or $x_0 < x_2 < x_1$.
We handle each case separately by choosing an appropriate join tree.
We pick $R_0 - R_1 - R_2$ for the first order, and $R_0 - R_2 - R_1$ for the second one.
In both cases, the inequalities can be eliminated in quasilinear time~\cite{tziavelis21inequalities},
yielding two distinct CQ–database pairs, one per order.
This elimination introduces a new variable for each join tree edge enforcing an inequality.
For example, for the first order, we obtain 
$Q_1(x_0, x_1, x_2, y, v_1, v_2) \allowbreak\datarule\allowbreak R_0'(x_0, v_1),\allowbreak R_1'(x_1, y, v_1, v_2),\allowbreak R_2'(x_2, y, v_2)$,
where $v_1, v_2$ are fresh variables and $R_0', R_1', R_2'$ are larger than the original relations by logarithmic factors.
To see why the partial order decomposition is needed,
observe that $x_0 \leq \min(x_1, x_2)$ can alternatively be viewed as a conjunction of $x_0 \leq x_1$ and $x_0 \leq x_2$.
However, these inequalities cannot be eliminated together because there is no single join tree that allows placing both on corresponding edges.
\end{example}

First, we formalize when a partial order can be eliminated.

\begin{definition}
    Let $\pord$ be a strict partial order over the variables 
    of a CQ $Q$.
    We say that a join tree of $Q$ enforces $\pord$ if for every pair of variables such that $x_1<x_2$ in $\pord$ and there is no $x_3$ with $x_1<x_3<x_2$, 
    we have that $x_1$ and $x_2$ either appear in the same node or in neighboring nodes.
\end{definition}

\begin{lemma}[\cite{tziavelis21inequalities}]
\label{cor:enforce_order}
Let $\pord$ be a strict partial order over the variables of a CQ $Q$.
If there exists a join tree of $Q$ that enforces $\pord$,
then $P_\pord$ can eliminated from $Q$ in quasilinear time.
\end{lemma}

We cannot always simply use all total orders between the $X$ variables in which $x_0$ is the minimum, but we may need partial orders tailored to our query. As an example, to handle the query $Q(x_0, x_1, x_2, y_1, y_2) \allowbreak\datarule\allowbreak R_1(x_1, y_1),\allowbreak R_0(x_0, y_1, y_2),\allowbreak R_2(x_2, y_2)$ with
$x0\leq \min(x1,x2)$, we can compare $x_0$ directly with $x_1$ and $x_2$, but no join tree lets us compare $x_1$ with $x_2$.
Our challenge is therefore to carefully choose partial orders such that, in addition to having an enforcing join tree for each partial order, their combination covers all possibilities where $x_0$ is the minimum.
To formalize this notion, we say that a set of strict partial orders $\pordset$ forms a partition of a set of strict total orders $O$ if:
(1) For every $\pord \in \pordset$, all strict total orders that extend $\pord$ are in $O$;
(2) For every $o\in O$, there exists exactly one $\pord \in \pordset$ that can be extended to $o$.
We also say that $x_0$ is the minimum element in a strict partial order over a set $X \cup \{x_0\}$ if, for all $x\in X$, the strict partial order requires that $x_0<x$.

\begin{lemma}\label{lem:partitioning}
    Let $Q$ be a full acyclic CQ containing a set of variables $X$ and $x_0\notin X$, such that there is no chordless $k$-path between two variables in $X \cup \{x_0\}$ with $k\geq 3$. Then, there exists a set $\pordset$ of strict partial orders of $X \cup \{x_0\}$ such that:
    \begin{itemize}
        \item $\pordset$ constitutes a partition of the strict total orders of $X \cup \{x_0\}$ in which $x_0$ is the minimum.
        \item For every $\pord \in \pordset$, there exists a join tree of $Q$ that enforces $\pord$.
    \end{itemize}
\end{lemma}

To prove \Cref{lem:partitioning},
we describe a recursive algorithm that builds the required partition along with enforcing join trees.
The algorithm starts from a given join tree, but rearranges it to ensure that certain orders are enforced.
A \emph{rearrangement} of a join tree is a join tree with the same nodes, possibly connected differently.
One possible way of rearranging a given tree is making it as `shallow and wide' as possible; We call such a join tree \emph{maximally branching}.

\begin{definition}[Maximally-Branching Tree]
    A rooted join tree is \emph{maximally branching} if we cannot disconnect any node from its parent and connect it to a different ancestor to obtain a join tree.    
\end{definition}

Given an arbitrary join tree, we can easily turn it into maximally-branching by repeatedly fixing violations in a root-to-leaf order.

\begin{toappendix}
    \begin{observation}\label{obs:turn-max-branching}
    Any join tree admits a rearrangement into a maximally-branching join tree with the same root.
\end{observation}
\begin{proof}
    Repeatedly fix violations by handling all nodes in a root-to-leaf manner. When handling a node, disconnect it from its parent, and reconnect it to the highest possible ancestor (that preserves the join tree property).
    A node that is fixed will not create further violations because its ancestors do not change after fixing it.
\end{proof}
\end{toappendix}

\begin{figure}
    \centering
    \includegraphics[width=\linewidth]{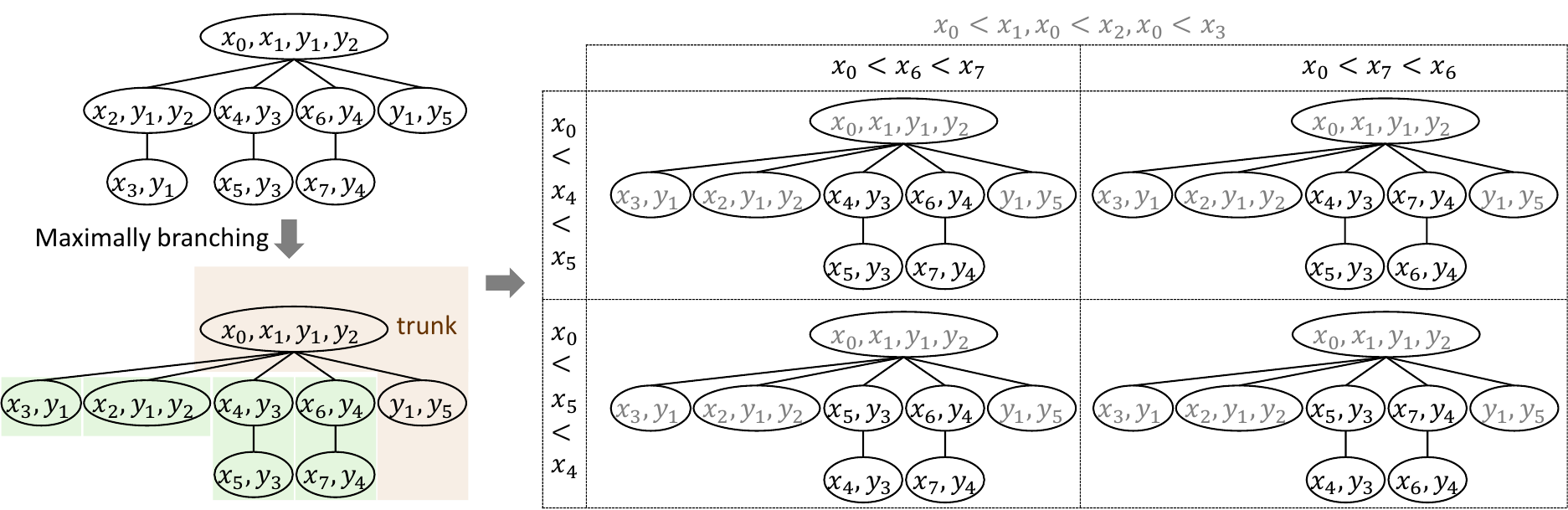}
    \caption{\Cref{ex:elim}: For an example join tree, partitioning $x_0 < \min \{ x_1, \ldots, x_7 \}$ into a set of partial orders and enforcing join trees. In the four resulting join trees, the common parts are shown in gray.}
    \label{fig:example_elim}
\end{figure}

To be able to capture all efficient cases, the algorithm works only with maximally-branching trees.
\Cref{alg:eliminate} shows the pseudocode, using the following notation:
Given a variable $x$,
$\neighbors{T}(x)$ denotes all variables that appear in a node of $T$ with $x$,
and $\highest{T}(x)$ the closest node to the root that contains $x$. This is uniquely defined because of the join-tree property.
Given a node $v$, 
$\subtree{T}{v}$ denotes the subtree of $T$ rooted in $v$,
and $\parent{T}(v)$ the parent node of $v$ in $T$.
Given a subtree $T'$ and a set $X$ of variables, $\restrict{X}{T'}$ denotes the subset of $X$ variables that appear in $T'$.
Finally, a union of graphs is the graph consisting of the union of the vertices and the union of the edges.

Nodes that contain $x_0$ are referred to as the \emph{trunk}.
Non-trunk nodes that have a parent in the trunk are called \emph{branch roots}, and the subtree rooted in a branch root is called a \emph{branch}.
The algorithm only considers branches that contain $X$ variables that do not appear in the trunk; the rest of the branches are not relevant to handle the predicate and are treated as part of the trunk.

Inequality variables that appear in the trunk can be directly compared with $x_0$ (lines~\ref{line:trunkvars}-\ref{line:trunk-inequalities}).
Non-trunk variables are partitioned according to the branch in which they appear (line~\ref{line:branchvars}), and each branch is considered independently.
For each branch, we consider all possibilities for the variable $x$ which is the minimum among the relevant variables in the branch (line~\ref{line:choose-branch-min}).
After recursing on the subtree, the tree is rearranged so that a node that contains $x$ is connected directly to the trunk, enforcing the inequality $x_0 < x$ (lines~\ref{line:branch-inequality}-\ref{line:branch-edge}).

\begin{example}
\label{ex:elim}
    \Cref{fig:example_elim} shows an example application of the algorithm, producing four join trees and associated orders.
    First, the tree is rearranged to become maximally branching.
    This puts $x_0$ and $x_3$ on neighboring nodes, which allows enforcing the inequality $x_0<x_3$ in addition to $x_0<x_2$.
    We remark that we cannot avoid this rearrangement and support these two inequalities as one branch because, while we can enforce the order $x_0 < x_2 < x_3$, we cannot enforce the order $x_0 < x_3 < x_2$.
    Also observe that $\{y_1, y_5\}$ does not contain any $X$ variables, and so it is considered part of the trunk. 
    The branch with inequality variables $x_4,x_5$ admits two potential orders, each one with a corresponding rearrangement of the tree.
    The same is true for the branch with $x_6,x_7$, resulting in four combinations when putting the branches together.
\end{example}

\begin{algorithm}
\footnotesize
\SetKwFunction{RecFun}{part\_min\_orders\_rec}
\SetKwProg{Fn}{Function}{:}{}

\textbf{Input}: rooted join tree $T$, variable $x_0$ in the root, subset $X$ of the variables of $T$, $x_0\notin X$\\
\textbf{Output}: a set of (partial order, enforcing join tree) pairs, as specified in \Cref{lem:alg-elimination-correct} \\

\Return $\mathrm{\RecFun}(T, x_0, X)$\;

\Fn{\RecFun{$T, x_0, X$}}{

    Rearrange $T$ so that it is maximally branching\;
    $\trunkvars=X\cap \neighbors{T}(x_0)$ \label{line:trunkvars} \algocomment{The $X$ variables that appear in the trunk}
    $\trunkpord=\{x_0<x\mid x\in\trunkvars\}$ \label{line:trunk-inequalities} \algocomment{The inequalities that can be enforced in the trunk} 
    $\branchRoots=\{\text{nodes } r\mid x_0\notin r, x_0\in\parent{T}(r), \restrict{(X\setminus \trunkvars)}{\subtree{T}{r}}\neq\emptyset\}$\;
    $\trunktree=$ copy of $T$ with the subtrees rooted in $\branchRoots$ removed \algocomment{The trunk}
    \For(\algocommentinline{Each branch is handled independently, setting inequalities between $X_r$ and $x_0$}){$r\in \branchRoots$ }
    {
        $X_r = \restrict{(X\setminus \trunkvars)}{\subtree{T}{r}}$ \label{line:branchvars} \algocomment{The relevant $X$ variables in this branch} 
        Initialize $\resr{r}$ as empty\;
        \For(\algocommentinline{Consider the case where $x$ is the minimum among $X_r$ \phantom{aaaaaaaaaaaaaaaaaaaaaaaaaa}}){$x\in X_r$ \label{line:choose-branch-min}} 
        { 
        $S = \mathrm{\RecFun}(\subtree{T}{r}, x, X_r)$\; \label{line:recursion}
        Add $x_0<x$ to every order in $S$\; \label{line:branch-inequality}
        Add edge $(\parent{T}(r), \highest{T}(x))$ to each tree in $S$ \label{line:branch-edge} \algocomment{Reconnect the subtree to the trunk} 
        Insert $S$ to $\resr{r}$\;
        }
    }
    Initialize \res{} as empty\;
    
    \For(\algocommentinline{All combinations of order-tree from each branch}){$(\pord_1,T_1), \ldots, (\pord_b, T_b)$ $\in$ $\bigtimes_{r \in \branchRoots} \resr{r}$}
    {
        Insert $(\trunkpord \cup \pord_1 \cup \ldots \cup \pord_b, \trunktree \cup T_1 \cup \ldots \cup T_b)$ to \res\; \label{line:res-combination}
    }
    
    \Return \res\;

}

\caption{Partition Min Orders}
\label{alg:eliminate}
\end{algorithm}

\begin{toappendix}
In order to prove \Cref{lem:branching-paths}, we need two auxiliary lemmas for rearranging join trees.

\begin{lemma}\label{lem:rearrange}
    Let $v_1,v_2,...,v_k$ be a path on a join tree. If $v_1\cap v_2\subseteq v_k$, then removing the edge $v_1-v_2$ and adding the edge $v_1-v_k$ results in a join tree.
\end{lemma}
\begin{proof}
    We use the characterization of a join tree by the property that every variable appears in a (connected) subtree. Consider a variable $u$. If it does not occur in both $v_1$ and $v_2$, removing the edge $v_1-v_2$ does not disconnect the subtree of nodes that contain $u$, and the property still holds. Otherwise, removing $v_1-v_2$ disconnects that subtree. Since $v_1\cap v_2\subseteq v_k$, we have that $u$ occurs in $v_k$. Thus, both endpoints of the path $v_1,v_2,...,v_k$ contain $u$. Since the original tree is a join tree, all variables on that path contain $u$. Adding the edge $v_1-v_k$ therefore reconnects the two parts of the subtree containing $u$.
\end{proof}

The next lemma shows why maximally branching trees are useful.
Paths on these trees imply paths between variables in these nodes with limited appearances in other parts of the trees.

\begin{lemma}\label{lem:branching-paths}
For a maximally-branching join tree, the following hold:
    \begin{enumerate}
    \item Given a non-root node $v_p$ and its child $v_c$, we have that $v_c$ and $v_p$ share a variable that does not appear in any ancestors of $v_p$.\label{prop:mbranch-parent}
    \item Given a non-root node $v_a$ and its descendant $v_d$ such that $x_a$ and $x_d$ are variables in $v_a$ and $v_d$ respectively, there is a path from $x_a$ to $x_d$ such that the intermediate variables in the path do not appear in ancestors of $v_a$.\label{prop:mbranch-ancestor}
    \end{enumerate}
\end{lemma}
\begin{proof}
Assume by way of contradiction that Property~(\ref{prop:mbranch-parent}) does not hold, and denote by $v_g$ the parent of $v_p$.
Then, every variable in $v_p\cap v_c$ appears in an ancestor of $v_p$, and due to the join-tree property, it also appears in $v_g$.
We get that $v_p\cap v_c\subseteq v_g$, and by \Cref{lem:rearrange}, it is possible to disconnect $v_c$ from $v_p$ and connect it to $v_g$, contradicting the fact that the join-tree is maximally branching.
To prove Property~(\ref{prop:mbranch-ancestor}), apply Property~(\ref{prop:mbranch-parent}) to every parent and child on the path from $v_a$ to $v_d$.
\end{proof}
\end{toappendix}

The following lemma establishes the correctness of the algorithm.
The fact that the tree is maximally branching and that there are no long chordless paths between variables in $X\cup\{x_0\}$ is used as part of proving that the constructed trees are indeed join trees.

\begin{lemmarep}\label{lem:alg-elimination-correct}
    Consider a rooted join-tree $T$, a variable $x_0$ that appears in its root, and a subset $X$ (not containing $x_0$) of the variables in the tree, and assume there is no chordless $k$-path between two variables in $X\cup \{x_0\}$ with $k\geq 3$.
    \Cref{alg:eliminate} returns a set of pairs $(\pord,t)$ such that:
    (1) $\pord$ is a strict partial order over $X \cup \{x_0\}$;
    (2) $t$ is a rearrangement of $T$;
    (3) $t$ enforces $\pord$; and
    (4) The set of all returned partial orders constitutes a partition of the strict total orders of $X \cup \{x_0\}$ in which $x_0$ is the minimum.
\end{lemmarep}
\begin{appendixproof}

We proceed by induction.
The base case is that every $X$ variable appears in a node with $x_0$.
In this case, there are no branches. The algorithm returns one pair $(\trunkpord,\trunktree)$ with $\trunkpord$ being the strict partial order comprising $x_0<x$ for all $x\in X$. Indeed, $\{\trunkpord\}$ constitutes the trivial partition, $\trunktree=T$ is a trivial rearrangement, and it enforces $\trunkpord$.

Next, consider the general case.
The algorithm partitions the set of target variables to trunk variables (line~\ref{line:trunkvars}) and non-trunk variables according to the branch in which they appear (line~\ref{line:branchvars}).
Different iterations of line~\ref{line:branchvars} produce disjoint sets $X_r$ since the variables do not appear in the trunk, and each variable needs to belong to a connected subset of join-tree nodes.
Trunk variables are explicitly required to be larger than $x_0$ (line~\ref{line:trunk-inequalities}).
For each branch, we separate into cases according to the minimum variable in the branch (line~\ref{line:choose-branch-min}). Then, the algorithm ensures that $x_0$ is smaller than all branch variables by adding the inequality stating that $x_0$ is smaller than the branch minimum (line~\ref{line:branch-inequality}).
Since the target variables in each branch are disjoint, the constructed inequalities consistently form a strict partial order (there are no cycles of inequalities).
Thus, the returned orders are valid and they form the required partition.

We now claim that within each returned pair, the tree enforces the order. Trunk inequalities (line~\ref{line:trunk-inequalities}) are enforced within trunk nodes. The inequalities in the result of the recursive call are enforced by the subtree given by the recursive call (line~\ref{line:recursion}), and the added inequalities with the branch minimum variables (line~\ref{line:branch-inequality}) are enforced by the added edges (line~\ref{line:branch-edge}) since $x_0\in\parent{R}(r)$ and $x\in\highest{T}(x)$.

It is left to prove that the trees returned by the algorithm are indeed join trees (satisfying the join-tree property).
We will show that $\parent{T}(r)\cap r\subseteq \highest{T}(x)$.
According to \Cref{lem:rearrange}, removing the edge $(\parent{T}(r),r)$ and adding the edge $(\parent{T}(r),\highest{T}(X))$ gives a valid join tree. Replacing the edges within the branch with the rearrangement given by the recursion also preserves the join tree property.

Assume by way of contradiction that $\parent{T}(r)$ and $r$ share a variable $y$ that is not in $\highest{T}(x)$. 
We will show that in that case, a chordless path of length 3 or more must exist.
Since the tree is maximally branching, according to \Cref{lem:branching-paths}, $\highest{T}(x)$ and its parent $v_p$ have a variable $y'$ in common that does not appear above the parent.
According to \Cref{lem:branching-paths} again, there is a path from $r$ to $v_p$ that uses only variables that do not appear in $\parent{T}(r)$, which implies a path from $y$ to $y'$ where no variables other than $y$ appear in $\parent{T}(r)$. Take a chordless shortening of this path from $y$ to $y'$ (it may be the case that $y$ and $y'$ are neighbors, $r=v_p$, and the chordless path is of length one), and extend it to the path: $x_0-y-...-y'-x$. We claim that this extended path is chordless too.
Since $x_0$ does not appear below $\parent{T}(r)$, and no variables other than $y$ appear in $\parent{T}(r)$, no chords involve $x_0$. Since $y'$ appears in $\highest{T}(x)$ and the path $y-...-y'$ is chordless, no variables on this path appear in $\highest{T}(x)$, and since $x$ does not appear above $\highest{T}(x)$, we have that $x$ is not part of any chords either. 
Thus, $x_0-y-...-y'-x$ is a chordless path of length $3$ or more, which is a contradiction.
\end{appendixproof}

\begin{proof}[Proof of \Cref{lem:partitioning}]
Due to \Cref{lem:alg-elimination-correct}, 
\Cref{lem:partitioning} is obtained by taking a join tree $T$ for $Q$ rooted in a node containing $x_0$ and then applying \Cref{alg:eliminate} with $T$, $x_0$, and $X$.
\end{proof}

We are now in position to show the positive side of \Cref{thm:remove-min} for full CQs.

\begin{proofsketch}
To be able to work with strict inequalities and strict total orders,
we transform the domains of the variables to ensure they are disjoint:
For a variable $v$, we replace the domain element $c$ with the tuple $(c,v)$.
That is, comparisons of distinct values $c$ are handled normally,
and ties across different variables are handled by comparing the variable identifiers, with $x_0$ having the smallest identifier.
After this transformation, we apply \Cref{lem:partitioning} and \Cref{cor:enforce_order}.
\end{proofsketch}

\subsection{Predicate Elimination with Projection}
\label{sec:elimin-ext}

We now consider the case where the min-predicate contains existential variables.
This case is challenging even for an acyclic free-connex query, since we cannot reduce it to a full query via the standard reduction (\Cref{lem:existential-elimination}), unless we first eliminate the inequalities that involve existential variables.
Fortunately, we show that such an elimination is always possible for all acyclic queries.

\begin{toappendix}
    
\begin{lemma}
\label{lem:max_value}
Let $Q$ be a full acyclic CQ, $y$ one of its variables, $\alpha$ one of its atoms, and $D$ a database.
We can compute in linear time a mapping from the assignments of $\var(\alpha)$ to
the maximum $y$ value that appears together with that assignment in $Q(D)$.
\end{lemma}
\begin{proof}
We use \Cref{lem:semiring-agg} with: 
(1) the max-tropical semiring $(\N \cup \{ -\infty \}, \max, +, -\infty, 0)$,
(2) an input function $\val$ that associates the tuples of one of the relations containing $y$ with the $y$ value,
and all other tuples with $0$, and
(3) a join tree with $\alpha$ as the root node.
Consequently, the aggregation result with $\otimes$ (within a query answer) will always be equal to $y$,
and the aggregation result with $\oplus$ (across query answers) will return their maximum $y$ value.
Since the tuples $t$ of the root relation correspond to the different assignments of $\var(\alpha)$,
the computed values $\agg(t)$ give the desired result.
\end{proof}

\end{toappendix}

\begin{lemmarep}\label{lem:existential-elimination}
Let $Q$ be an acyclic CQ, $y$ an existential variable, and $P \equiv (x \leq y)$ an inequality predicate. Given a database $D$, we can construct in linear time a database $D’$ by removing tuples from $D$ such that $(Q \wedge P)(D) = Q(D')$.
\end{lemmarep}
\begin{proofsketch}
    For each tuple in an atom containing $x$, we compute using \Cref{lem:semiring-agg} the maximum $y$ value $y_\text{max}$ that appears with it in query answers, and we discard the tuples with $x > y_\text{max}$.
\end{proofsketch}
\begin{appendixproof}

    Pick an arbitrary atom $\alpha$ of the query that contains the variable $x$, and let $R_\alpha$ be the corresponding relation.
    Also, let $Q_f$ be the same as $Q$, but with all variables free.
    We apply \Cref{lem:max_value} on $Q_f$, $P$, and $D$ to compute,
    for each tuple $t \in R_\alpha$,
    the maximum $y$ value $m_y(t)$ that agrees with $t$ in $Q_f(D)$.
    Then, we scan $R_\alpha$ and remove all tuples that do not satisfy $t[x] \leq m_y(t)$.\footnote{\label{footnote:other-elim}If we want to handle a predicate of the form $x<\min X$ with $x\notin X$ instead, we need to use this lemma with a predicate of the form $x<y$, and here we remove all tuples that do not satisfy the strict inequality $t[x] < m_y(t)$.}
    We claim that the obtained database $D'$ satisfies our lemma.

    First, we show that for $q \in (Q \wedge P)(D)$, we have that $q \in Q(D')$.
    Since $q$ is a query answer to $Q$ together with $P$, there is an assignment $q_e$ for existential variables such that $P$ is satisfied for $q \cup q_e$.
    This means that $q_f = q \cup q_e$ is in $(Q_f \wedge P)(D)$ and that $q_f$ assigns values $x_f$ and $y_f$ to $x$ and $y$ such that $x_f \leq y_f$.
    Let $t$ be the restriction of $q_f$ on $\var(\alpha)$.
    The value $m_y(t)$ computed by \Cref{lem:max_value} for $t$ will have to be at least $y_f$, and so $t$ will not be removed from $R_\alpha$.
    Since only tuples from the relation $R_\alpha$ are removed by our process, it follows that $q_f$ will appear in $Q_f(D')$ and $q$ in $Q(D')$.

    Conversely, we show that for $q \in Q(D')$, we have that $q \in (Q \wedge P)(D)$.
    For $q$ to be in $Q(D')$, its restriction $t$ on the $\alpha$ variables must satisfy $t[x] \leq m_y(t)$, so $x \leq y$ in $q$, and $q$ satisfies $P$.
\end{appendixproof}

We apply this lemma to remove all existential variables from a predicate.
Given a predicate $P \equiv (x \leq \min X)$, we define the restriction of $P$ to a set of variables $S$ to be the predicate $(x\leq \min (X\cap S))$ if $x\in S$ and $X\cap S\neq\emptyset$, or `True' otherwise.
Given a CQ $Q$ with a predicate $P$, we define their restriction to free variables to be $Q'\wedge P'$, where $Q'$ and $P'$ are restrictions of $Q$ and $P$ to the free variables of $Q$. 

\begin{lemmarep}\label{lem:restrict-predicate-to-free}
    Let $Q$ be an acyclic free-connex CQ, $X$ a subset of its variables, and $P \equiv (x \leq \min X)$ a predicate. Let $Q'\wedge P'$ be the restriction of $Q\wedge P$ to the free variables.
    Given a database $D$, it is possible to build in linear time a database $D'$ such that $(Q\wedge P)(D)=(Q'\wedge P')(D')$.
\end{lemmarep}
\begin{proofsketch}
    The min-predicate is treated as a conjunction $\bigwedge_{x_i \in X} x \leq x_i$,
    and the inequalities that involve existential variables are iteratively removed by \Cref{lem:existential-elimination}.
\end{proofsketch}
\begin{appendixproof}
    We treat $P=(x \leq \min X)$ as the conjunction of inequalities $\bigwedge_{x_i \in X} x \leq x_i$. Then, $P'$ is a conjunction of those inequalities that contain two free variables.
    Given a database $D$, we eliminate all inequalities $P_e$ that involve at least one existential variable to get a database $D''$.
    More specifically, if $P_i$ is the $i$-th inequality in the conjunction $P_e$ (in arbitrary order), we eliminate $P_i$ from $Q$ using \Cref{lem:existential-elimination} and obtain a database $D_i$.
    Then, we repeat the process for the next inequality starting from $Q$ and $D_i$.
    At the end of each iteration, $Q(D_i)$ satisfies the conjunction up to $P_i$.
    The last database we construct is $D''$, for which $(Q\wedge P')(D'') = (Q\wedge P' \wedge P_e)(D) = (Q \wedge P)(D)$.
    We then use \Cref{lem:cq-free-restrict} on $Q$ and $D''$ to eliminate the existential variables from $Q$, and we get $D'$ such that $Q(D'')=Q'(D')$.
    Overall, $(Q\wedge P)(D)=(Q\wedge P')(D'')=(Q'\wedge P')(D')$.
\end{appendixproof}

By using \Cref{lem:restrict-predicate-to-free} to restrict the query to its free variables and applying the elimination for full queries (\Cref{sec:elimination}), we complete 
the proof of the positive side of \Cref{thm:remove-min}.

\section{Direct Access Algorithm}\label{sec:da_min_alg}

To establish the 
positive side of \Cref{thm:min-da},
we present an algorithm called $\MINDA$ that relies on the elimination algorithm of \Cref{thm:remove-min}.
By \Cref{lem:cq-free-restrict}, it is enough to consider full CQs.

\begin{lemmarep}\label{thm:ranked-da-pos}
    Given a full acyclic CQ $Q$, a subset $X$ of free variables such that there is no chordless $k$-path between two variables in $X$ with $k \geq 3$, and a database $D$,
    direct access according to $\min{X}$ is possible with quasilinear preprocessing time and logarithmic access time.
\end{lemmarep}
\begin{appendixproof}
    The lemma is proved via the $\MINDA$ algorithm discussed in \Cref{sec:da_min_alg}.

    The preparation steps take linear time,
    materializing a constant number of new relations to eliminate self-joins,
    and modifying the domain values in a single pass.
    Elimination of the min-predicates (\Cref{thm:remove-min}) takes quasilinear time and there are $\bigO(1)$ queries in the resulting set $S$. 
    For each such query, the secondary direct access structure is built in quasilinear time using the algorithm of Carmeli et al.~\cite{carmeli23tractable} for a partial lexicographic order that includes a single variable.
    
    To construct the direct access array, we first consider each CQ in $S$ separately and compute $\qeq$ for each domain value $\minval$ of the variable $x$ designated to be the minimum. 
    Since the CQ is full and acyclic, we can compute these counts in linear time using \Cref{lem:semiring-agg} with: 
    (1) the counting semiring $(\N, \times, +, 1, 0)$,
    (2) an input function $\val$ that assigns $1$ to all tuples in the database,
    and (3) a join tree where the root is a node that contains only the variable $x$.
    So far, we computed pairs of the form $(\minval,\qeq)$. We sort these pairs according to $\minval$ and compute $\qps$ using a prefix sum. Then, we add the identifier $\qid$ and merge the sorted arrays coming from the different queries in $S$. Once we have the merged array sorted by $\minval$ and then $\qid$, we compute $\tps$ using a prefix sum. All these operations can be done in linear time.

    Let us explain the correctness of the algorithm when accessing an index $k$ using the entry $i$. The number of query answers that have a strictly smaller minimum value (after the domain modification from the first paragraph) is $\tps(i)$. To break ties between answers with the same minimum value (that also come from the same query $\qid(i)$), we have to access the answer with index $k-\tps(i)$ out of these answers. This is the answer with index $k-\tps(i)+\qps(i)$ out of the answers to the CQ $\qid(i)$ because we have to ``skip'' $\qps(i)$ answers within $\qid(i)$ to arrive at the target minimum value. In effect, our tie-breaking scheme orders the answers to $Q$ first by the identifier $\qid(i)$, and then by the tie-breaking scheme of the direct access structure of $\qid(i)$.
\end{appendixproof}

Our $\MINDA$ algorithm that proves \Cref{thm:ranked-da-pos} works as follows.

\paragraph{Preparation Steps} 
We start by removing any self-joins and modifying the database so that two variables are never assigned the same value, as in the proof of \Cref{thm:remove-min}.
For every variable $y$, the domain element $c$ is replaced with the pair $(c,y)$. Here, it does not matter how we break ties, and we can use an arbitrary order between the variable identifiers. From now on, in each query answer, the minimum $X$ value can be uniquely attributed to one variable $x$.

\paragraph{Partition \& Eliminate}
Next, $\MINDA$ constructs a collection of query-database pairs
that form a partition of the query answers. 
In each query-database pair, there is one designated variable that always gets assigned the minimum value out of $X$.
More specifically, for each $x \in X$, we apply \Cref{thm:remove-min} on $Q\wedge (x \leq \min X)$, and union all the resulting sets of query-database pairs into a set $S$.
In the following, we refer to $S$ as a set of queries and implicitly assume that each CQ is associated with its own database.

\paragraph{Secondary DA Structures}
For each query in $S$, we build a direct access structure that orders its answers by the minimum $X$ value. Since the minimum is always obtained by a designated variable $x$, it is enough to sort by $x$. Sorting by $x$ can also be seen as a partial lexicographic order $\langle x \rangle$. Since this lexicographic order has only one variable and the query is full acyclic, a structure with this order and logarithmic-time access can be built in quasilinear time~\cite{carmeli23tractable}.

\paragraph{Min-DA Array}
Next, we build an array of linear size in the input that represents all sorted query answers.
Each entry represents a set of answers that come from the same CQ in $S$ and have the same minimum value. An entry is comprised of five elements:
    \begin{itemize}
        \item $\minval$: A value of $\min X$.
        \item $\qid$: A query identifier for the queries in the set $S$.
        \item $\qeq$: The number of answers to $\qid$ in which $\min X$ is $\minval$.
        \item $\qps$: The number of answers to $\qid$ in which $\min X$ is strictly smaller than $\minval$.
        \item $\tps$: The sum of $\qeq$ of all preceding entries in the array.
    \end{itemize}
    The entries are sorted lexicographically by $\minval$ and then by $\qid$ (and only then $\tps$ is computed).
    As we explain in the appendix, this array can be built in quasilinear time, using \Cref{lem:semiring-agg} to compute $\qeq$ and using prefix sums to compute $\qps$ and $\tps$.

\paragraph{Access Phase}
When accessing a target index $k$, we first find the largest entry $i$ such that $k\geq \tps(i)$. Then, we access the answer with index $k-\tps(i)+\qps(i)$ in the query $\qid(i)$ and return its projection to the variables of $Q$.
The relevant entry can be found in logarithmic time using binary search, and the modified target index can be accessed within the CQ $\qid(i)$ in logarithmic time using the corresponding secondary DA structure.

\begin{figure}
    \centering
    \includegraphics[width=0.85\linewidth]{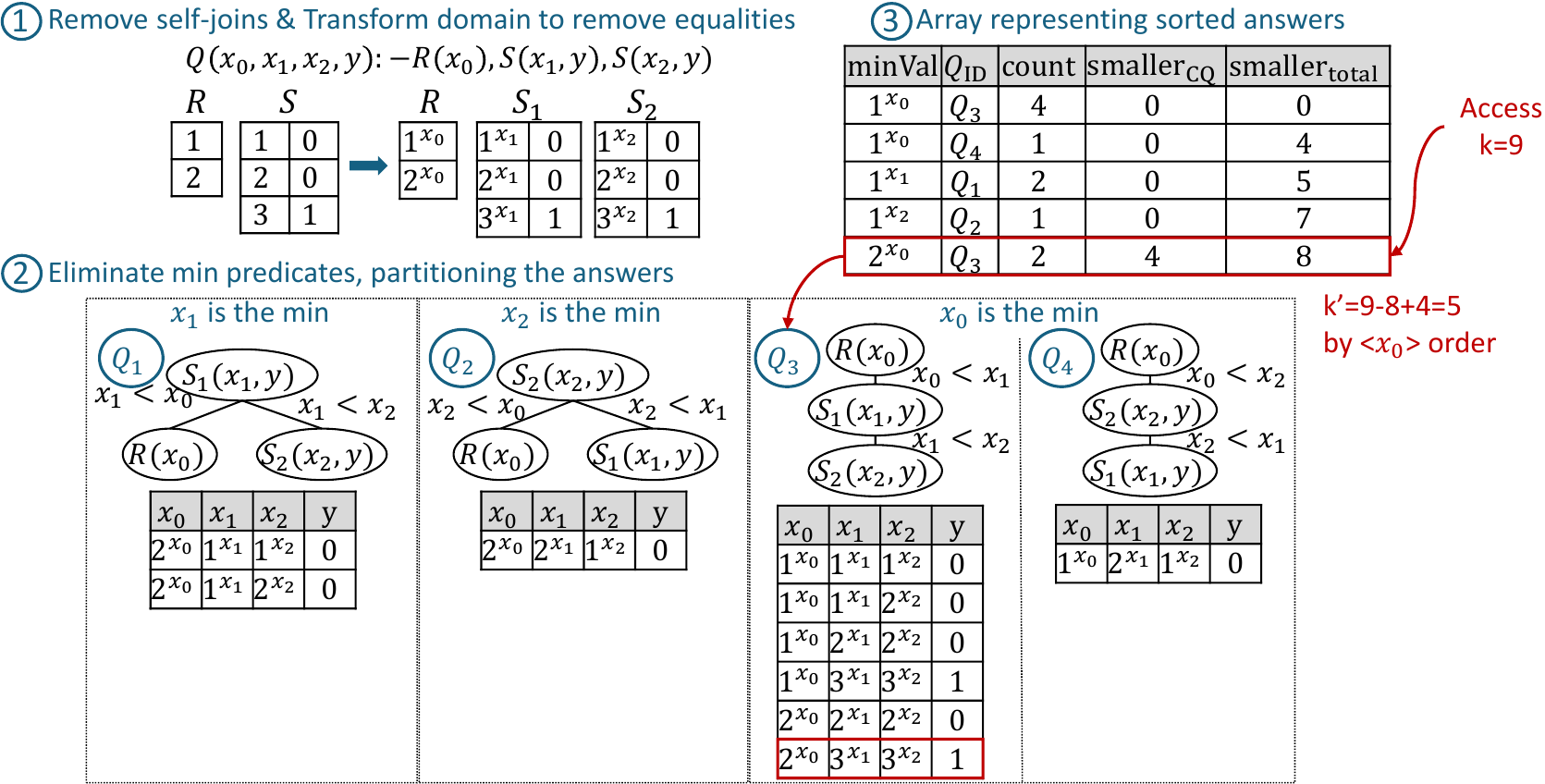}
    \caption{Example direct access for the full CQ with body $R(x_0),R(x_1,y),R(x_2,y)$ by $\min\{x_0,x_1,x_2\}$. For visualization purposes, the transformed domain values are shown with superscript (e.g., $1^{x_0}$ instead of $(1, x_0)$).}
    \label{fig:da_example}
\end{figure}

\begin{example}
\Cref{fig:da_example} shows an example execution of our algorithm. After removing self-joins and transforming the domain so that two distinct variables are never equal, the ranking function $\min\{x_0, x_1, x_2\}$ produces three cases, with each variable taking on the minimum value.
For each case, we invoke our elimination algorithm (\Cref{sec:elimination}).
For $x_1$ and $x_2$, each elimination results in a single CQ, while the $x_0$-minimum predicate is eliminated into 2 CQs (and corresponding databases).
The figure shows the query answers that each query produces, which are not actually materialized by the algorithm.
Instead, the shown inequalities are eliminated according to the enforcing join tree~\cite{tziavelis21inequalities}.
For example, the inequalities of $Q_3$ are handled by introducing two fresh variables $v_1, v_2$ and relations $R'(x_0, v_1), S_1'(x_1, y, v_1, v_2), S_2'(x_2, y, v_2)$.
The array representing the sorted answers is shown on the top right. An access call with index $k=9$ (with zero indexing) is handled by binary search on $\tps$ to find the largest value smaller than $k$, which directs us to the relevant query $Q_3$.
Next, we need to find answer number $k-\tps=1$ out of the answers to $Q_3$ with $\minval$ $2^{x_0}$, which is answer number $1+\qps=5$ out of all answers to $Q_3$. 
A secondary direct access structure by $\langle x_0 \rangle$ on $Q_3$ gives us the answer $(x_0, x_1, x_2, y) = (2, 2, 2, 0)$ after dropping the inequality variables $v_1, v_2$ and reverting back to the original domain.
\end{example}

\section{Enumeration}\label{sec:enum}

\subsection{Enumeration with min-predicates}\label{sec:pred-enum}

We now turn our attention to enumeration for acyclic free-connex CQs with a min-predicate $x \leq \min X$. Using \Cref{lem:restrict-predicate-to-free}, it is enough to consider full queries.
Based on our results so far, one strategy is to first eliminate the min-predicate, and then apply a standard enumeration algorithm to the resulting query.
However, an efficient such elimination is only possible under the restrictive no-chordless-$k$-path condition of \Cref{thm:remove-min}.
For queries where this condition does not hold, does enumeration fundamentally require higher complexity, or do we need a different approach? As an example, the query $Q(x_0, u, v, x_1, x_2) \allowbreak \datarule \allowbreak R_0(x_0,u), \allowbreak R_1(u,v), \allowbreak R_2(v,x_1), \allowbreak R_2(x_1,x_2)$ with $x_0 \leq \min(x_1, x_2)$ does not admit efficient trimming. It is also not supported by the known enumeration algorithm for CQs with inequality predicates~\cite{wang2022comparisons}.
We show that this query, as well as all full acyclic CQs with min predicates, admits efficient enumeration.

\begin{lemmarep}\label{lem:enum-predicate-pos}
    Let $Q$ be a full acyclic CQ, $X$ a subset of its variables, and $P \equiv (x_0 \leq \min X)$ a predicate. Then, enumerating the answers of $Q\wedge P$ is possible with linear preprocessing and constant delay.
\end{lemmarep}
\begin{proofsketch}
We take a join tree rooted in an atom containing $x_0$.
For each fact in every node of the join tree, its \emph{threshold} is the maximum of $\min X'$ over all answers to the subtree rooted at that node that use this fact, where $X'$ is the subset of $X$ appearing in the subtree.
We compute the threshold of all facts using \Cref{lem:semiring-agg} with the max-min semiring. 
Enumeration proceeds top-down in the join tree, using only facts whose threshold is at least the current $x_0$ value.
\end{proofsketch}
\begin{appendixproof}
    Take a join tree for $Q$ rooted in an atom containing $x_0$.
    For every fact of every join-tree node, we define the \e{threshold} to be the maximum of $\min X'$ over all answers to the subtree rooted in this node that use this fact, where $X\setminus{x_0}$ is the subset of $X$ relevant to this subtree. 
    We compute the threshold for all facts using \Cref{lem:semiring-agg} with: 
    (1) the max-min semiring $(\N \cup \{ -\infty, +\infty \}, \max, \min, -\infty, +\infty)$, and
    (2) an input function $\val$ that associates the tuples of each relation with the minimum value among the $X\setminus{x_0}$ variables it contains (and $+\infty$ if none of these variables appears). 
    
    Notice that there are no answers to $Q\wedge P$ that assign $x_0$ with a value $c$ and use a fact $f$ if $c$ is larger than the threshold for $f$.
    Indeed, every answer to $Q$ must in particular induce an answer to the subtree; if the threshold is smaller than $c$, this means that, for every answer in the subtree, the minimum of some of the $X$ variables is smaller than $c$, so no answer in the subtree can be extended to satisfy the predicate when $x_0$ is assigned $c$.
    Thus, the assignment to $x_0$ can be used as an upper bound on the threshold values for facts we should consider.
    All we need to do at this point is to efficiently generate answers without considering facts whose threshold is below the $x_0$ value we consider.

    After computing all thresholds, our enumeration algorithm is a simple
    modification of the standard enumeration algorithm for acyclic joins~\cite{Berkholz20tutorial}.
    The latter works as follows.
    We start with a semi-join reduction bottom-up in the join tree, as in the Yannakakis algorithm~\cite{Yannakakis}.
    Now all facts that can be reached with top-down traversals participate in at least one query answer.
    To facilitate joins, in the preprocessing phase we also 
    split every relation into buckets such that facts in the same bucket agree on the assignments to the variables that this node shares with its parent. 
    The enumeration phase follows a nested-loop join, where we
    iterate over the join tree top-down (i.e., considering a node only after considering its parent). The loop iterates over all facts in the root relation, and while considering a specific fact, we go over all facts in a bucket of a child node that matches the currently considered assignment to the join variables. When all nodes are handled, the considered facts form an answer, and all answers are obtained this way. 
    
    In our variation of this algorithm, we sort every bucket in increasing order according to the threshold values. To only explore facts whose threshold is at least the $x_0$ value, we just need to stop the iteration over each bucket when reaching a threshold value smaller than $x_0$.\footnote{\label{footnote:other-predicate-enum}If we want to prove this lemma for a predicate of the form $\min X>x$ with $x\notin X$ instead, we would stop the iteration when reaching a threshold value smaller \emph{or equal} to the $x_0$ value.}

    From the discussion above, it is clear that this procedure does not miss answers: we already argued that the combination of $x_0$ assignments and facts that we do not consider in this process does not lead to answers to $Q\wedge P$.
    It is also clear that every answer we produce is a valid answer: when we produce an answer, we have selected one fact for each atom, where all facts are compatible with each other; for every $x_i\in X\setminus\{x_0\}$, we chose a fact $f$ containing $x_i$ and verified that the assignment for $x_0$ is at most the threshold for $f$, which is in turn at most the assignment we chose for $x_i$ (because all answers in this subtree assign this same value for $x_i$ and assignments to other variables can only make the threshold smaller).
    It is left to argue that the delay obtained using this procedure remains constant. We need to claim that, in every non-root bucket we consider, there is at least one fact whose threshold is at least the $x_0$ value we consider. Consider the fact chosen in the parent node that leads to this bucket. Since we explore the bucket, it means that the threshold of this parent fact is at least $x_0$. By definition of this threshold, there is an answer to the subtree of this parent for which the minimum $X$ assignment is the threshold. This answer uses some fact from the child bucket, and this child fact has a threshold at least the $x_0$ value.
\end{appendixproof}

This proves the positive side of \Cref{thm:enum-pred}. 
The negative side follows from the known hardness of self-join-free CQs that are not acyclic free-connex~\cite{bdg:dichotomy,braultbaron13thesis} since we can set the variable domains such that $x_0$ is always smaller than the other variables and the predicate always holds.

\subsection{Ranked Enumeration with Min Ordering}\label{sec:ranked-enum}

Next, we prove \Cref{thm:ranked-enum-dichotomy}.
The hardness of enumeration with arbitrary order extends to ranked enumeration: self-join-free CQs that are not acyclic free-connex do not admit efficient enumeration with any ordering, assuming Hyperclique and BMM.
In contrast, acyclic free-connex queries admit efficient ranked enumeration by $\min X$ for any $X$.
Known algorithms~\cite{deep25ranked,tziavelis25anyk} achieve logarithmic delay using appropriate priority queue structures.
We improve this to constant delay using a simpler algorithm.
By \Cref{lem:cq-free-restrict}, it is enough to consider full queries.

\begin{lemmarep}\label{lem:ranked-enum-pos}
    Let $Q$ be a full acyclic CQ and $X$ a subset of its variables. 
    $Q$ admits enumeration according to $\min{X}$ with linear preprocessing time and constant delay.
\end{lemmarep}
\begin{proofsketch}
    We enumerate in parallel the answers ranked by $x_i$ for each $x_i\in X$. At each step, we print the minimum answer across them, and we ignore answers already printed. The delay between answers can be regularized using the ``Cheater's Lemma''~\cite{carmeli2021UCQs}.
\end{proofsketch}
\begin{appendixproof}
    First, notice that, given $x_i\in X$, it is possible to enumerate the answers of $Q$ in increasing order of the $x_i$ values with linear preprocessing and constant delay. To this end, it suffices to use the well-known Yannakakis algorithm on a join tree rooted in an atom containing $x_i$ where the relation for this atom is sorted by the $x_i$ values.
    We use $S_i$ to denote a structure implementing this algorithm that allows: checking if there are more answers, reading the current answer (we are allowed to ask to read this answer several times), and advancing to the next answer.

    To print the answers according to the minimum $X$ value,
    we use these structures in parallel for every $x_i\in X$. For every structure $S_i$ that has remaining answers, let $a_i$ be its current answer, and denote by $v_i$ the value that $a_i$ assigns $x_i$; we choose $r$ such that $v_r$ is the minimum out of all $v_i$ values; in case of ties, we take the smallest $i$ that gives the minimum. Then, we print $a_r$ and advance $S_r$ to the following answer.

    This process produces every answer $|X|$ times. If we only consider the first time each answer is produced, the answers are sorted in increasing order of $\min X$ as required.
    To avoid the duplicates, we choose one structure to be in charge of printing each answer: the structure that proposes the answer when it is first chosen. 
    If $x_i$ is assigned the minimum value out of $X$ in an answer,
    then $S_i$ is responsible for that answer;
    in case of ties, we choose the smallest $i$.
    Before a structure prints an answer, we check if it is assigned to it, and if not, we simply ignore the answer, advance the structure to the following answer, and repeat the process. 

    The above fix removes duplicates but may compromise the delay. 
    We claim that the algorithm runs in linear partial time. That is, the time from the beginning of the execution up to the $k$-th answer is $O(|D|+k)$.
    It is known that a linear partial time algorithm can be transformed into a linear preprocessing and constant delay algorithm~\cite{carmeli2021UCQs}\footnote{The transformation into a constant delay algorithm works by withholding answers and delaying the time in which they are produced to regularize the delay. This transformation shows that, in a sense, linear partial time is not worse than linear preprocessing and constant delay. However, it has two disadvantages~\cite[Section~2.3]{tziavelis25anyk}: (1) some answers may be printed later than in the original algorithm; (2) a large space consumption may be required to store answers seen but not yet printed. Thus, while this transformation is useful to prove that we can get the theoretical guarantee we seek, from a practical point of view, it makes sense to keep the original algorithm without the constant delay guarantee.}.
    From our choice of structure in charge, we have that whenever an answer is skipped, it was already printed through a different structure. So, at any point during the execution, if the number of answers printed is $k$, the number of skips performed is at most $k(|X|-1)$. Since $|X|$ is a constant and every skip requires a constant number of operations, the time elapsed from the beginning of the enumeration process until we print the $k$-th answer is at most $O(k)$. Together with the preprocessing, this gives $O(|D|+k)$, so linear partial time as needed.
\end{appendixproof}

\section{Lower Bounds}\label{sec:hardness}

In this section, we prove lower bounds matching our upper bounds for the problems of predicate elimination, counting, and direct access.
Enumeration is not mentioned in this section since our dichotomies use known lower bounds (see \Cref{sec:enum}).
As predicate elimination and direct access can be used for counting, we will use the hardness of counting to prove hardness for all three tasks. 

\subsection{Counting and Predicate Elimination}\label{sec:count-neg}

We first show the hardness of counting answers with identical end-points of a long chordless path.

\begin{lemmarep}\label{lem:counting-path-hard}
    Let $Q$ be a self-join-free CQ, and assume there is a chordless $k$-path between the variables $x_1$ and $x_2$ with $k\geq 3$. Then, we cannot count the number of answers to $Q$ that satisfy $x_1=x_2$ in quasilinear time, assuming \hyperclique{}.
\end{lemmarep}
\begin{proofsketch}
Given an undirected graph, we construct a database for which an assignment to the chordless $k$-path in the query corresponds to a $3$-path in the graph with end-points $x_1$, $x_2$. A $3$-path with equal end-points represents a graph triangle. Since detecting triangles in the graph is not possible in quasilinear time assuming \hyperclique, we conclude that we cannot count the answers with $x_1 = x_2$ in that time.
\end{proofsketch}
\begin{appendixproof}
    We build on the assumption that it is not possible to detect triangles in a graph in quasilinear time in the number of its edges, which is a direct consequence of \hyperclique{}. So, assuming we are given an (undirected) graph with edges $E$, we will show that quasilinear time counting for $Q$ with $x_1 = x_2$ entails that we can detect a triangle in quasilinear time.

    Consider a chordless $k$-path $x_1,y_1,\ldots,y_{k-1},x_2$ of $Q$ with $k\geq 3$, and denote the atoms on this path by $R_1(\vec{v_1}),\ldots,R_k(\vec{v_k})$.
    We construct a database over which this path corresponds to a $3$-path.
    Consider the atom $R_1(\vec{v_1})$.
    For every edge $e=\{a,b\}\in E$ with $a\neq b$, we define the function $f_1^e$ by $f_1^e(x_1)=a$, $f_1^e(y_1)=b$, and all other variables map to the constant $\bot$. Then, we add the fact $R_1(f_1^e(\vec{v_1}))$ to our database. We proceed similarly for two more atoms on the path.
    We define the functions $f_2^e$ and $f_k^e$ by $f_2^e(y_1)=a$, $f_2^e(y_2)=b$, $f_k^e(y_{k-1})=a$, $f_k^e(x_2)=b$ and all other variables map to the constant $\bot$ in both functions. We add the facts $R_2(f_2^e(\vec{v_2}))$ and $R_k(f_k^e(\vec{v_k}))$ to our database.
    Lastly, we define the function $f^e$ to assign any variable of the path with $a$ and all other variables with $\bot$. 
    For every atom $S(\vec{u})$ of the query other than $R_1(\vec{v_1})$, $R_2(\vec{v_2})$ and $R_k(\vec{v_k})$, add the fact $S(f^e(\vec{u}))$.

    Over this construction, the first two atoms on the path, as well as the last one, correspond to graph edges. We also have that $y_i=y_{i+1}$ for all $i \geq 2$, and all variables that are not on the path always correspond to the constant $\bot$. Thus, the query answers correspond to the $3$-paths in the graph, and the end-points of such paths are the assignments to $x_1$ and $x_2$. The answer assignments that satisfy $x_1=x_2$ therefore correspond to the triangles in the input graph, and counting these answers determines the existence of triangles (if there are answers, there are triangles). Since detecting triangles is not possible in quasilinear time (assuming \hyperclique{}), and the construction here can be done with quasilinear time, we conclude that counting the query answers that satisfy $x_1 = x_2$ cannot be done in quasilinear time.  
\end{appendixproof}

In order to use \Cref{lem:counting-path-hard} to prove the hardness of queries with minimum predicates, we show next that counting with minimum predicates allows counting with equality predicates.

\begin{lemmarep}\label{lem:counting-min-to-equality}
    Let $Q$ be a self-join-free CQ, $X$ a subset of its variables, $x_1,x_2\in X$ free variables, and $P \equiv (x_1 \leq \min X)$ a predicate.
    If $Q\wedge P$ admits counting in $O(t)$ time, then we can count the number of answers to $Q$ that satisfy $x_1=x_2$ in $O(t)$ time.
\end{lemmarep}
\begin{proofsketch}
    To compute the number of answers that satisfy $x_1=x_2$, we subtract the number of answers that satisfy $x_1<x_2$ from those that satisfy $x_1\leq x_2$. Each of the latter two quantities can be computed using the predicate $x_1 \leq \min X$ after a linear-time domain transformation.
\end{proofsketch}
\begin{appendixproof}
    Consider a database $D$, and recall that we assume that domain values are integers.
    Given that $\mathtt{cnt}(Q(D)\mid x_1=x_2)=\mathtt{cnt}(Q(D)\mid x_1\leq x_2)-\mathtt{cnt}(Q(D)\mid x_1<x_2)$, it is enough to compute the last two terms.

    To compute the number of answers where $x_1\leq x_2$, we replace every value $v$ in the domain of $x_1$ or $x_2$ by $(1,v)$, and we replace every value $v$ in the domain of the other variables by $(2,v)$. This way, other variables are always assigned a larger value than that of $x_1$ and $x_2$, and the answers in which $x_1\leq x_2$ are exactly those in which $x_1 \leq \min X$. Thus, $Q\wedge P$ computes exactly this count.

    To compute the number of answers where $x_1<x_2$, we replace every value $v$ in the domain of $x_1$ by $(1,v+1)$.
    We replace every value $v$ for $x_2$ by $(1,v)$ and every value $v$ for the other variables by $(2,v)$ as before. In this case, $Q\wedge P$ computes the answers in which $x_1+1\leq x_2$, which are exactly those where $x_1<x_2$.
    \footnote{\label{footnote:other-predicate-neg}If we want to prove \Cref{lem:counting-min-to-equality} for a predicate of the form $\min X>x$ with $x\notin X$ instead, we can compute the two quantities similarly.
    For the answers where $x_1<x_2$, replace every value $v$ in the domain of $x_1$ or $x_2$ by $(1,v)$, and replace every value $v$ for the other variables by $(2,v)$.
    For the answers where $x_1\leq x_2$, replace every value $v$ in the domain of $x_1$ by $(1,v-1)$, replace every value $v$ for $x_2$ by $(1,v)$, and replace every value $v$ for the other variables by $(2,v)$.}
\end{appendixproof}

We now deduce the hardness of counting for queries with minimum predicates.

\begin{lemmarep}\label{lem:counting-hard-side}
    Let $Q$ be a self-join-free CQ, $X$ a subset of its variables, and $P \equiv (x \leq \min X)$ a predicate.  
    Counting $Q\wedge P$ is not possible in quasilinear time if $Q$ is not acyclic free-connex, or if $x$ is free and $Q$ has a chordless $k$-path between $x$ and another free variable in $X$ with $k\geq 3$, assuming \hyperclique{} and \seth{}.
\end{lemmarep}
\begin{proofsketch}
    If $Q$ is not acyclic free-connex, we cannot efficiently count the answers to $Q$~\cite{mengel2025lower}, so it is enough to transform the domain such that $x$ is always the smallest and $P$ always holds.
    Otherwise, \Cref{lem:counting-path-hard} implies that we cannot count the answers in which the end-points of the chordless path are equal, and then \Cref{lem:counting-min-to-equality} implies that we cannot count for $Q\wedge P$.
\end{proofsketch}
\begin{appendixproof}
    The first case is that $Q$ is not acyclic free-connex. In this case, we set the domain of $x$ to be smaller than that of the other variables, e.g. by replacing every value $v$ in the domain of $x$ by $(1,v)$ and every value $v$ in the domain of the other variables by $(2,v)$. This way, the predicate $x\leq\min X$ is always satisfied, and $Q\wedge P$ gives the same answers as $Q$. Quasilinear time counting for a CQ that is not free-connex acyclic contradicts \hyperclique{} or \seth{}~\cite{mengel2025lower}.

    The second case is that $Q$ has a chordless $k$-path between the free variable $x$ and another free variable $y\in X$ with $k\geq 3$. If $Q\wedge P$ admits quasilinear-time counting, we can use \Cref{lem:counting-min-to-equality} to count the answers in which $x=y$ in quasilinear time, which contradicts \hyperclique{} according to \Cref{lem:counting-path-hard}.
\end{appendixproof}

\Cref{lem:counting-hard-side} gives the negative side of the counting dichotomy (\Cref{thm:count-min-dichotomy}). As a full acyclic elimination allows efficient counting, this directly implies the negative side of the predicate elimination dichotomy (\Cref{thm:remove-min}).
We next discuss direct access.

\subsection{Direct Access}\label{sec:neg-da}

We first recall that direct access can be used for counting by binary searching for out-of-bounds.

\begin{observationrep}[\cite{nofar22random}]\label{obs:DA-for-counting}
    Direct access for any query $Q$ (in any order) with $t_p$ preprocessing time and $t_a$ access time can be used to count the answers to $Q$ over a database $D$ in time $O(t_p+t_a\log|D|)$.
\end{observationrep}
\begin{appendixproof}
    This can be done using binary search since the access calls are required to return out-of-bound when the requested index is larger than the number of answers.
    The number of access calls required is logarithmic in the number of answers. In data complexity, the number of answers is bounded by a polynomial in the database size, so the number of access calls is also logarithmic in the size of the input database.
\end{appendixproof}

We will rely on \Cref{obs:DA-for-counting} to show the hardness of queries that are not acyclic free-connex. For the other queries, we will rely on the hardness of counting answers with identical end-points of a long chordless path, as given by \Cref{lem:counting-path-hard}. For this, we need to show that direct access according to an order defined by a minimum condition allows counting with equality predicates.

\begin{lemmarep}\label{lem:DA-to-equality-counting}
    Let $Q$ be a self-join-free CQ that admits direct access according to $\min X$ with $t_p$ preprocessing time and $t_a$ access time, and let $x_1,x_2\in X$. 
    Then, we can count the number of answers to $Q$ that satisfy $x_1=x_2$ over a database $D$ in time $O(t_p+|\dom(D)|\cdot t_a\log|D|)$.
\end{lemmarep}
\begin{proofsketch}
    Given a domain value $v$, denote by $C_{\min}(v)$, $C_1(v)$, and $C_2(v)$ the number of query answers that satisfy $v=\min\{x_1,x_2\}$, $v=x_1<x_2$, and $v=x_2<x_1$ respectively.
    We want to compute the number of answers that satisfy $x_1=x_2$, which is equal to $\sum_{v \in \dom(D)}{\left(C_{\min}(v)-C_1(v)-C_2(v)\right)}$.

    For $C_{\min}$, we transform the domain such that $\min{X}=\min\{x_1,x_2\}$. We binary search using direct access calls to find all the answer indices where the value of $\min\{x_1,x_2\}$ changes, and thus compute $C_{\min}(v)$ for every domain value $v$.
    For $C_{1}$, we further transform the domain such that the values of $x_2$ are disjoint and slightly decreased compared to those of $x_1$. This way, when $\min{X}$ is given by $x_1$, we know that $x_1<x_2$ in the original domain. For every domain value for $x_1$, we binary search for the number of answers in which this is the $\min{X}$ value.
    The computation of $C_2$ is symmetrical.
\end{proofsketch}
\begin{appendixproof}
    Consider a database $D$, and recall that we assume that the domain values are integers.
    Given a domain value $v$, we set the following notation.
    \begin{align*}
        C_{\min}(v)&=\mathtt{cnt}(Q(D) | \min\{x_1,x_2\}=v)\\
        C_{eq}(v)&=\mathtt{cnt}(Q(D) | x_1=x_2=v)\\
        C_1(v)&=\mathtt{cnt}(Q(D) | x_1=v \wedge x_2>v)\\
        C_2(v)&=\mathtt{cnt}(Q(D) | x_2=v \wedge x_1>v)
    \end{align*}
    Notice that $C_{\min}(v)=C_{eq}(v)+C_{1}(v)+C_{2}(v)$, and the count we want to compute is 
    \[\mathtt{cnt}(Q(D) | x_1=x_2) =
    \sum_{v \in \dom(D)}{C_{eq}(v)}=
    \sum_{v \in \dom(D)}{\left(C_{\min}(v)-C_1(v)-C_2(v)\right)}\]

    First, let us describe how to compute $C_{\min}(v)$ for every domain value $v$.
    Given a database $D$, we replace every value $u$ in the domain of $x_1$ or $x_2$ by $(1,u)$, and we replace every value $u$ in the domain of the other variables by $(2,u)$. This way, other variables are always assigned a larger value than that of $x_1$ and $x_2$, and we have that $\min{X}=\min\{x_1,x_2\}$. Thus, we are given a direct access algorithm where the answers are sorted by $\min\{x_1,x_2\}$.
    We can binary search using the direct access calls to find all of the answer indices where the value of $\min\{x_1,x_2\}$ changes, and thus compute $C_{\min}(v)$ for every domain value in total time $O(t_p+|\dom(D)|\cdot t_a\log|D|)$.

    Now let us describe how to compute $C_1(v)$ for every domain value $v$. 
    We construct a copy of the database where the values of $x_2$ are slightly decreased compared to those of $x_1$, and the values of the other variables are always higher. More precisely, we replace every value $u$ of $x_2$ with $(1,2u-1)$, we replace every value $u$ of $x_1$ with $(1,2u)$ and every value $u$ of the other variables with $(2,u)$.
    Consider a value $v$ for which we want to compute $C_1(v)$.
    We use the same binary search procedure as before to compute the number of answers in which $\min{X}=(1,2v)$.
    As the assignment to the second component of $x_2$ is always odd, the minimum in these answers is always given by $x_1$. These answers include exactly those where $2x_1=2v$ and $2x_2-1>2v$, which are the answers in which $x_2>v=x_1$, so this indeed computes $C_1(v)$.
    Performing this computation for every domain value $v$ takes $O(t_p+|\dom(D)|\cdot t_a\log|D|)$ time. The symmetric procedure can be done to compute $C_2(v)$ for every domain value $v$.
\end{appendixproof}

We can now prove the hardness of direct access as given by the negative side of \Cref{thm:min-da}.

\begin{lemmarep}\label{lem:DA-hard-side}
    Let $Q$ be a self-join-free CQ and $X$ a subset of its variables. 
    Direct access according to $\min{X}$ is not possible with quasilinear preprocessing time and polylogarithmic access time if $Q$ is not acyclic free-connex, or if it has a chordless $k$-path between two variables in $X$ with $k\geq 3$, assuming \hyperclique{} and \seth{}.
\end{lemmarep}
\begin{proofsketch}
    If $Q$ is not acyclic free-connex, it does not admit efficient counting~\cite{mengel2025lower}, and no efficient direct access by \Cref{obs:DA-for-counting}.
    Otherwise, \Cref{lem:counting-path-hard} implies that we cannot count the answers in which the end-points of the chordless path are equal, and \Cref{lem:DA-to-equality-counting} implies that we cannot have efficient direct access by $\min{X}$.
\end{proofsketch}
\begin{appendixproof}
    The first case is that $Q$ is not acyclic free-connex. In this case, $Q$ admits no quasilinear-time counting~\cite{mengel2025lower}, and \Cref{obs:DA-for-counting} implies that it has no direct access with quasilinear preprocessing time and polylog access time (regardless of the answer order).
    The second case is that $Q$ has a chordless $k$-path between two variables $x_1,x_2\in X$ with $k\geq 3$. 
    According to \Cref{lem:DA-to-equality-counting}, direct access according to $\min{X}$ allows to count the answers to $Q$ that satisfy $x_1=x_2$. Since such counting cannot be done in quasilinear time according to \Cref{lem:counting-path-hard}, we conclude that such direct access cannot be done with quasilinear preprocessing and polylog access time.
\end{appendixproof}

\section{Conclusion}\label{sec:conclusion}

We developed an algorithm that eliminates a min-predicate from a CQ, producing a set of full acyclic CQs with disjoint answer sets.
Building on this elimination, we presented dichotomies for a variety of problems:
elimination, counting, and unranked direct access with a min-predicate,
as well as direct access by min ranking.
The tractable cases for these problems are exactly those where the query is acyclic free-connex and there is no long chordless path between two free variables relevant for the min condition.
This is in contrast to similar enumeration problems, where tractability requires no additional condition on top of the query being acyclic free-connex.
A visual summary of our results is shown in \Cref{fig:overview}, with formal statements presented in \Cref{sec:results}.

As the proofs for several query-answering tasks involving min-predicates or rankings turned out to be interdependent, we are hopeful that, in addition to the dichotomy results in this paper, our work can serve as a toolbox for exploring a broader landscape of queries and tasks.
There are several natural generalizations of our results that are worth considering. 
Since our results are restricted to a single min or max predicate,
we would like to extend them to multiple such predicates in a conjunction.
This is equivalent to a conjunction of inequalities, which have been considered~\cite{tziavelis21inequalities,wang2022comparisons}, but are not entirely understood for all queries and tasks.
It would also be interesting to consider a wider class of predicates involving min or max functions and more complex rankings.
For example, while all known results for direct access assume a simple ranking function based on sum, min, max, or a lexicographic order, it is intriguing to explore their combinations,
such as ordering first by the min of some variables and then, in case of a tie, by their lexicographic order.

\bibliographystyle{ACM-Reference-Format}
\bibliography{references.bib}

\end{document}